\documentclass[11pt]{article}
\usepackage[a4paper,left=1.0in,right=1.0in,top=1.1in,bottom=1.3in]{geometry}

\usepackage{setspace}
\usepackage{microtype}

\usepackage{authblk}
\usepackage{xcolor}
\usepackage{etex,ifthen}
\usepackage{etoolbox}
\usepackage{graphicx}
\usepackage{rotating}
\usepackage{tikz}
\usepackage{pgfplots}\pgfplotsset{compat=1.14}
\usepackage{tikz-cd}
\usepackage{url}
\usepackage{calc}
\usepackage[font={small,it}]{caption}
\usepackage{complexity}
\usepackage{alphalph}
\usepackage{amsthm}
\usepackage{thmtools}
\usepackage{eqparbox}
\usepackage{afterpage}
\usepackage{amsfonts,amsmath,amssymb}
\usepackage{mathrsfs}
\usepackage{algorithm}
\usepackage{algpseudocode}
\usepackage{bm}
\usepackage{mdframed}
\usepackage{listings}
\usepackage[english]{babel}
\usepackage{hyperref}
\usepackage{verbatim}
\usepackage[all]{xy}

\usetikzlibrary{chains,fit}
\usetikzlibrary{shapes.geometric}
\usetikzlibrary{matrix,positioning,calc}
\usetikzlibrary{decorations.markings}
\usetikzlibrary{decorations.pathmorphing}
\tikzstyle{vertex}=[circle, draw, inner sep=0pt, minimum size=6pt]
\tikzstyle{svertex}=[circle, draw, inner sep=0pt, minimum size=3pt]
\tikzstyle{dvertex}=[circle, draw, inner sep=0pt, minimum size=9pt]
\tikzstyle{vertbox}=[draw, inner sep=0pt, minimum size=8pt]
\newcommand{\vertex}{\node[vertex]}
\newcommand{\svertex}{\node[svertex]}
\newcommand{\dvertex}{\node[dvertex]}

\newcommand{\oset}[3][0ex]{
    \mathrel{\mathop{#3}\limits^{
		\vbox to#1{\kern-2\ex@
		\hbox{$\scriptstyle#2$}\vss}}}}

\newcommand{\cupdot}{\mathbin{\mathaccent\cdot\cup}}
\newcommand{\floor}[1]{\left\lfloor #1 \right\rfloor}
\newcommand{\ceil}[1]{\left\lceil #1 \right\rceil}
\newcommand{\interval}{\mathcal{I}}
\newcommand{\Cost}{\mathcal{C}}
\newcommand{\paths}{\mathcal{P}}
\newcommand{\cA}{\mathcal{A}}
\newcommand{\cB}{\mathcal{B}}
\newcommand{\cF}{\mathcal{F}}
\newcommand{\cX}{\mathcal{X}}

\newcommand{\rev}{\operatorname{rev}}
\newcommand{\Gmn}[2][]{G_{m#1,n,#2}}
\newcommand{\GL}{G^L}
\newcommand{\GM}{G^M}
\newcommand{\GR}{G^R}
\newcommand{\linkgadget}{\mathbf{L}}

\newcommand{\num}{\mathsf{num}}
\newcommand{\den}{\mathsf{den}}
\newcommand{\linkedge}{\mathtt{link}}
\newcommand{\tS}{\mathtt{S}}
\newcommand{\sols}{\mathfrak{L}}
\newcommand{\compl}{\varphi}
\newcommand{\complpoly}[1]{\varphi^{\operatorname{deg}(#1)}}
\newcommand{\complmult}[1]{\varphi^{\operatorname{lin}(#1)}}
\newcommand{\pl}{\operatorname{pl}}
\newcommand{\npl}{\operatorname{npl}}
\newcommand{\gr}{\operatorname{gr}}

\newcommand{\Zn}{\mathbb{Z}_n}
\newcommand{\Znhat}{\hat{\mathbb{Z}}_n}
\newcommand{\Expt}{\mathbb{E}}
\newcommand{\var}{\mathbf{var}}
\newcommand{\conv}{\mathbf{conv}}
\newcommand{\Grid}{\Upsilon}
\newcommand{\pred}{\mathsf{pred}}

\newcommand{\vlambda}{\vec{\lambda}}
\newcommand{\va}{\vec{a}}
\newcommand{\vx}{\vec{x}}
\newcommand{\vy}{\vec{y}}
\newcommand{\vz}{\vec{z}}

\newcommand*\quotefont{\fontfamily{LinuxLibertineT-LF}} 
\newcommand*\quotesize{60}
\newcommand*{\openquote}
    {\tikz[remember picture,overlay,xshift=-4ex,yshift=-2.5ex]
    \node (OQ) {\quotefont\fontsize{\quotesize}{\quotesize}\selectfont``};\kern0pt}
\newcommand*{\closequote}[1]
    {\tikz[remember picture,overlay,xshift=4ex,yshift={#1}]
    \node (CQ) {\quotefont\fontsize{\quotesize}{\quotesize}\selectfont''};}
\colorlet{shadecolor}{cyan!10}
\newcommand*\shadedauthorformat{\emph}
\newcommand*\authoralign[1]{
    \if#1l \def\authorfill{}\def\quotefill{\hfill}
    \else \if#1r \def\authorfill{\hfill}\def\quotefill{}
    \else \if#1c \gdef\authorfill{\hfill}\def\quotefill{\hfill}
    \else\typeout{Invalid option}
    \fi \fi \fi}
{\authoralign{#1}
\ifblank{#2}
    {\def\shadequoteauthor{}\def\yshift{-2ex}\def\quotefill{\hfill}}
    {\def\shadequoteauthor{\par\authorfill\shadedauthorformat{#2}}\def\yshift{2ex}}
\begin{snugshade}\begin{quote}\openquote}
{\shadequoteauthor\quotefill\closequote{\yshift}\end{quote}\end{snugshade}}

\declaretheorem[numberlike=equation]{Theorem}
\declaretheorem[numberlike=equation]{Lemma}
\declaretheorem[numberlike=equation]{Corollary}
\declaretheorem[numberlike=equation]{Conjecture}
\declaretheorem[numberlike=equation]{Proposition}
\declaretheoremstyle[bodyfont=\it,qed=$\lozenge$]{defstyle}
\declaretheorem[numberlike=equation,style=defstyle]{Definition}
\declaretheorem[numberlike=equation]{Claim}
\declaretheorem[numberlike=equation]{Fact}

\makeatletter
\patchcmd{\ALG@step}{\addtocounter{ALG@line}{1}}{\refstepcounter{ALG@line}}{}{}
\newcommand{\ALG@lineautorefname}{Line}
\makeatother

\addto\extrasenglish{%
}

\hypersetup{
	colorlinks,
	linkcolor=red!75!black,
	citecolor=blue!75!black,
	urlcolor=green!75!black
}

\title{Parametric Shortest Paths in Planar Graphs}

\author{Kshitij Gajjar\thanks{\texttt{kshitij.gajjar@tifr.res.in}} }
\author{Jaikumar Radhakrishnan\thanks{\texttt{jaikumar@tifr.res.in}}}

\affil{Tata Institute of Fundamental Research, Mumbai, India}

\begin{document}
	
	
	\maketitle
	
	
	
	
	\begin{abstract}
	We construct a family of planar graphs $\{G_n\}_{n\geq 4}$, where $G_n$ has $n$ vertices including a source vertex $s$ and a sink vertex $t$, and edge weights that change linearly with a parameter $\lambda$ such that, as $\lambda$ varies in $(-\infty,+\infty)$, the piece-wise linear cost of the shortest path from $s$ to $t$ has $n^{\Omega(\log n)}$ pieces. This shows that lower bounds obtained earlier by Carstensen (1983) and Mulmuley \& Shah (2001) for general graphs also hold for planar graphs, thereby refuting a conjecture of Nikolova (2009).
	
    Gusfield (1980) and Dean (2009) showed that the number of pieces for every $n$-vertex graph with linear edge weights is $n^{\log n + O(1)}$. We generalize this result in two ways. (i) If the edge weights vary as a polynomial of degree at most $d$, then the number of pieces is $n^{\log n + (\alpha(n)+O(1))^d}$, where $\alpha(n)$ is the slow growing inverse Ackermann function. (ii) If the edge weights are linear forms of three parameters, then the number of pieces, appropriately defined for $\mathbb{R}^3$, is $n^{(\log n)^2+O(\log n)}$.
	\end{abstract}
	
	
	
	
	\section{Introduction}
	
	We consider the following \emph{parametric shortest path problem} on graphs. The input is a directed acyclic graph with two special vertices $s$ and $t$. The edges have weights that vary linearly with a real-valued parameter $\lambda$, that is, the weight of each edge $e$ is a function of the form $w_e(\lambda)=a_e \lambda + b_e$, for some real numbers $a_e$ and $b_e$. The cost of an $s$-$t$ path $P$ is the sum of the weights of the edges on it; therefore this cost is also a linear function of $\lambda$ of the form $\Cost(P)(\lambda)= \sum_{e\in P} a_e \lambda + \sum_{e\in P} b_e$. The cost of the shortest $s$-$t$ path is then given by
    \[ \Cost(\lambda) = \min_{P}\ \Cost(P)(\lambda),\]
    where $P$ ranges over all $s$-$t$ paths; this function is the piece-wise linear lower envelope (\autoref{fig:planar3x3}) of the linear costs provided by the $s$-$t$ paths. The main object of our investigation is the number of pieces in this envelope. This quantity is of interest in several applications; in particular, determining this quantity for \emph{planar graphs} has been a subject of several studies.
    
    
    \begin{figure}
	\begin{center}
	\begin{tikzpicture} [scale=0.90]
        \draw [->, gray!65, line width=1mm] (7,-1) to (14.5,-1);
        \node at (15,-1) {\large{\bf $\lambda$}};
        \draw [->, gray!65, line width=1mm] (7,-1) to (7,5);
        \node at (7,5.5) {\large{\bf cost of path}};
        
        \draw[teal] (6.8,-0.7)--(12,4.5); \draw[step=.3, yshift=-1cm, densely dotted] (6,0) grid (6.6,0.6); \draw[very thick, red] (6,-1)--(6.6,-1)--(6.6,-0.4);
        \draw[teal] (6.8,0.2)--(14,2); \draw[step=.3, yshift=-0.1cm, densely dotted] (6,0) grid (6.6,.6); \draw[very thick, red] (6,-0.1)--(6.3,-0.1)--(6.3,0.2)--(6.6,0.2)--(6.6,0.5);
        \draw[teal] (6.8,1.32)--(14,0.6); \draw[step=.3, yshift=+0.12cm, densely dotted] (6,0.9) grid (6.6,1.5); \draw[very thick, red] (6,1.02)--(6,1.32)--(6.6,1.32)--(6.6,1.62);
        \draw[teal] (6.8,3.4)--(14,-0.2); \draw[step=.3, yshift=+0.1cm, densely dotted] (6,3) grid (6.6,3.6); \draw [very thick, red] (6,3.1)--(6.3,3.1)--(6.3,3.7)--(6.6,3.7);
        
        \draw[magenta, very thick] (7,-0.5)--(8,0.5)--(10,1)--(12,0.8)--(14,-0.2);
        
        \draw[teal] (6.8,2.3)--(11,4.5); \draw[step=.3, yshift=-0.1cm, densely dotted] (6,2.1) grid (6.6,2.7); \draw [very thick, red] (6,2)--(6,2.3)--(6.3,2.3)--(6.3,2.6)--(6.6,2.6);
        \draw[teal] (6.8,4.5)--(14,3); \draw[step=.3, densely dotted] (6,4.2) grid (6.6,4.8); \draw [very thick, red] (6,4.2)--(6,4.8)--(6.6,4.8);
        
        \begin{scope} [xshift=-1.5cm]
        {
		    \vertex at (0,0) [fill=blue, label=below left:$s$] (v00) {};
		    \vertex at (0,2) [fill=blue] (v02) {};
		    \vertex at (0,4) [fill=blue] (v04) {};
			
		    \vertex at (2,0) [fill=blue] (v20) {};
		    \vertex at (2,2) [fill=blue] (v22) {};
		    \vertex at (2,4) [fill=blue] (v24) {};
			
		    \vertex at (4,0) [fill=blue] (v40) {};
		    \vertex at (4,2) [fill=blue] (v42) {};
		    \vertex at (4,4) [fill=blue, label=above right:$t$] (v44) {};
		}
		\end{scope}
		
		\begin{scope} [decoration={markings, mark=at position 0.6 with {\arrow[scale=2,>=stealth,gray]{>}}}]
		    \draw [postaction={decorate}] (v00)-- node[below, opacity=0.6] {\scriptsize{$a_1\lambda + b_1$}} (v20);
		    \draw [postaction={decorate}] (v20)-- node[below, opacity=0.6] {\scriptsize{$a_2\lambda + b_2$}}(v40);
		    \draw [postaction={decorate}] (v02)-- node[above, opacity=0.6] {\scriptsize{$a_3\lambda + b_3$}}(v22);
		    \draw [postaction={decorate}] (v22)-- node[below, opacity=0.6] {\scriptsize{$a_4\lambda + b_4$}}(v42);
		    \draw [postaction={decorate}] (v04)-- node[above, opacity=0.6] {\scriptsize{$a_5\lambda + b_5$}}(v24);
		    \draw [postaction={decorate}] (v24)-- node[above, opacity=0.6] {\scriptsize{$a_6\lambda + b_6$}}(v44);
		    
		    \draw [postaction={decorate}] (v00)-- node[opacity=0.6, rotate=90,yshift=0.25cm] {\scriptsize{$c_1\lambda + d_1$}}(v02);
		    \draw [postaction={decorate}] (v02)-- node[opacity=0.6, rotate=90,yshift=0.25cm] {\scriptsize{$c_2\lambda + d_2$}}(v04);
		    \draw [postaction={decorate}] (v20)-- node[opacity=0.6, rotate=90,yshift=0.25cm] {\scriptsize{$c_3\lambda + d_3$}}(v22);
		    \draw [postaction={decorate}] (v22)-- node[opacity=0.6, rotate=90,yshift=-0.25cm] {\scriptsize{$c_4\lambda + d_4$}}(v24);
		    \draw [postaction={decorate}] (v40)-- node[opacity=0.6, rotate=90,yshift=-0.25cm] {\scriptsize{$c_5\lambda + d_5$}}(v42);
		    \draw [postaction={decorate}] (v42)-- node[opacity=0.6, rotate=90,yshift=-0.25cm] {\scriptsize{$c_6\lambda + d_6$}}(v44);
		\end{scope}
	\end{tikzpicture}
	\caption{The figure on the left is a directed acyclic graph (horizontal edges go rightwards and vertical edges go upwards) with edge weights as linear functions of $\lambda$. The figure on the right plots $\lambda$ versus the (linear) costs of the $6$ different $s$-$t$ paths in the graph. The piece-wise linear lower envelope of this plot has $4$ pieces.}
	\label{fig:planar3x3}
	\end{center}
	\end{figure}
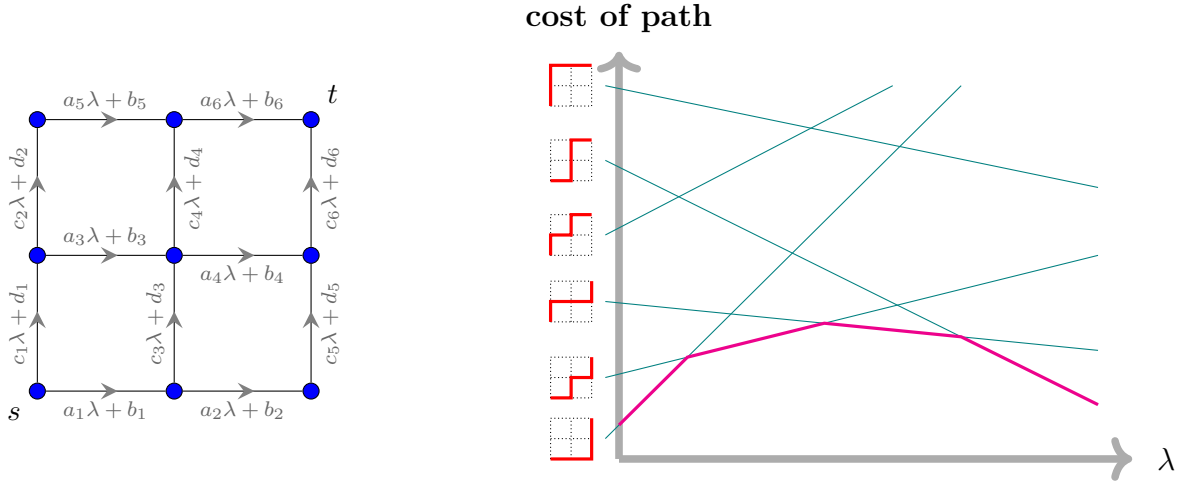
    
    Let the \emph{parametric shortest path complexity}, denoted by $\compl(n,\beta(n))$, be the maximum possible number of pieces in $\Cost(\lambda)$ for graphs on $n$ vertices, where the bit lengths of the coefficients in the weights of the edges are bounded by $\beta(n)$. 
    Let $\compl^{\pl}(n,\beta(n))$ be the parametric complexity when the graphs are restricted to be planar. We show the following.
    
    \begin{restatable}[Main result, lower bound on the parametric complexity of planar graphs]{Theorem}{thmmainresult} \label{thm:mainresult}
    \[ \compl^{\pl}(n, O((\log n)^3)) = n^{\Omega(\log n)}. \]
    \end{restatable}
    
    Similar results were known for general graphs. Carstensen~\cite{carstensen,carstensenthesis} showed that $\compl(n,\infty) = n^{\Omega(\log n)}$; her result was simplified and extended by Mulmuley \& Shah~\cite{mulmuleyshah}, who showed that $\compl(n, O((\log n)^3)) = n^{\Omega(\log n)}$.  For \emph{planar graphs,}  Nikolova~\cite[Conjecture 6.1.6]{nikolovathesis} conjectured that the complexity is bounded by a polynomial in $n$, that is, $\compl^{\pl}(n,\infty) = n^{O(1)}$. Our main result provides a strong (with bit length $O((\log n)^3)$) refutation of this conjecture.
    
    The above lower bounds are tight. Carstensen~\cite[Page 100]{carstensenthesis} presented a matching upper bound, $\compl(n,\infty) = n^{\log n+O(1)}$, which she attributed to Daniel Gusfield~\cite{gusfieldthesis,gusfield} (a similar argument, attributed to Brian Dean, was presented by Nikolova~\cite[Page 86]{nikolovathesis}). We generalize these upper bounds in two ways based on how the edge weights $w_e$ vary: (i) $w_e(\lambda)$ is a polynomial of degree at most $d$ in  $\lambda \in \mathbb{R}$, and (ii) $w_e(\lambda_1,\lambda_2,\lambda_3)=a_e\lambda_1+b_e\lambda_2+c_e\lambda_3$, where $(\lambda_1,\lambda_2,\lambda_3)\in\mathbb{R}^3$. A slightly different definition of $\compl$ is required for these generalizations.
    
    
    Consider a graph $G$ whose edge weights depend on a parameter $\lambda$ taking values in a set $\Lambda$. As $\lambda$ varies, the minimum cost $s$-$t$ path might vary. We say that a set $\paths$ of $s$-$t$ paths \emph{covers} $G$ if for every value $\lambda \in \Lambda$, the set $\paths$ contains a minimum cost $s$-$t$ path of $G$. We define the parametric shortest path complexity of $G$ as
    \[ \compl(G) = \min_{\paths:\,\paths \text{ covers } G} |\paths|.\]
    
    For the generalization (i), let $\complpoly{d}(n)$ be the maximum value of $\compl(G)$ as $G$ varies over $n$-vertex graphs whose edge weights are polynomials of degree at most $d$ in a parameter $\lambda \in \mathbb{R}$. We use the inverse Ackermann function (\autoref{def:ackermann}) to upper bound $\complpoly{d}(n)$.
    \begin{restatable}[Upper bound with univariate polynomial edge weights]{Theorem}{thmcomplpoly} \label{thm:complpoly}
    \[\complpoly{d}(n) = n^{\log n + (\alpha(n)+O(1))^d},\] where $\alpha(n)$ is the extremely slow growing inverse Ackermann function.
    \end{restatable}
    For the generalization (ii), let $\complmult{k}(n)$ be the maximum value of $\compl(G)$ as $G$ varies over $n$-vertex graphs whose edge weights have the form $w_e(\vlambda)=\va_e \cdot \vlambda$, where $\va_e\in\mathbb{R}^k,\vlambda\in\mathbb{R}^k$.
    \begin{restatable}[Upper bound with three-parameter linear edge weights]{Theorem}{thmcomplmult} \label{thm:complmult}
    \[ \complmult{3}(n) = n^{(\log n)^2+O(\log n)}. \]
    \end{restatable}
    
    \paragraph{Remarks:} (i) Note that upper bounds in~\autoref{thm:complpoly} and~\autoref{thm:complmult} grow only moderately (for small $d$). (ii) \autoref{thm:complmult} leads to the natural question whether similar bounds can be shown for $\complmult{k}(n)$ in general; unfortunately, our proof method fails when $k >3$. (iii) A bound of the form $\complmult{2}(n) = n^{\log n + O(1)}$ can be derived using our method; this implies Gusfield's bound cited above.

    \subsection{Significance of the main result}
    
    In this section, we present some consequences of our main result (\autoref{thm:mainresult}).
    
    \paragraph{\textit{PRAM lower bounds:}} From their result (that is, $\compl(n, O((\log n)^3)) = n^{\Omega(\log n)}$), Mulmuley \& Shah~\cite{mulmuleyshah} derived a lower bound on the running time of unbounded fan-in PRAMs with bit operations with a small number of processors solving the shortest path problem. \autoref{thm:mainresult} allows us to make a similar claim for planar graphs (see~\autoref{sec:mulmuley-shah-discussion} for a discussion on this).
    
    \begin{Claim}\label{thm:corollarymulmuley}
        There exist constants $\alpha>0, \epsilon>0$, and an explicitly described family of weighted planar graphs $\{G_n\}$ $(G_n$ has $n$ vertices, and the edge weights of $G_n$ are $O((\log n)^3)$ bits long$)$, such that for infinitely many $n$, every unbounded fan-in PRAM algorithm (without bit operations) with at most $n^\alpha$ processors requires at least $\epsilon \log n$ steps to compute the shortest $s$-$t$ path in $G_n$.
    \end{Claim}
    
    \paragraph{\textit{Weighted graph matching:}} Mulmuley \& Shah observed that their result for the shortest path problem yields the same lower bound for the \textsc{Weighted Graph Matching} problem~\cite[Corollary 1.1]{mulmuleyshah}. Our result extends this observation to planar graphs (see~\autoref{sec:weighted-graph-matching} for a proof of this reduction). Many graph problems are easier to solve for planar graphs than for general graphs; in particular, we note the $\NC$ algorithm for counting perfect matchings based on the work of Kasteleyn~\cite{kasteleyn} and Csanky~\cite{csanky}, and its remarkable recent application by Anari \& Vazirani~\cite{anarivazirani} (see also~\cite{sankowski}) to find perfect matchings in planar graphs. It is interesting that the lower bound for the \textsc{Weighted Graph Matching} problem derived by Mulmuley \& Shah continues to hold even when the input is restricted to be planar.
    
    \paragraph{\textit{Treewidth:}} Planar graphs have high (superpolynomial) parametric complexity (\autoref{thm:mainresult}). It is thus natural to ask what graph classes might have small parametric complexity. Since every planar graph can be embedded in a grid graph with a marginal increase in size, our lower bound holds for $n\times n$ grid graphs as well. We also explore the parametric complexity of $k\times n$ grid graphs when $k\ll n$ (see~\autoref{sec:thinvsthick}). Due to a result of Chekuri \& Chuzhoy~\cite{chekurichuzoy}, every graph of treewidth $k$ has an $\Omega(k^{1/99})\times\Omega(k^{1/99})$ grid minor. Thus, for large enough $k$, every graph of treewidth $k$ has parametric complexity $k^{\Omega(\log k)}$. In particular, if the graph class has $n$-vertex graphs whose treewidth grows as $k(n)=\exp(\omega(\sqrt{\log n}))$, then its parametric complexity grows superpolynomially. On the other hand, our construction shows that for every $k(n)=\omega(\log n)$, there are planar graphs with treewidth $k(n)$ and parametric complexity $n^{\omega(1)}$; in the reverse direction, it can be shown that $n$-vertex graphs of treewidth $k$ have parametric complexity $n^{O(k)}$~\cite{pranabendu}.
    
    \paragraph{\textit{Minimum weight $s$-$t$ cut:}} Our planar graphs are built such that $s$ and $t$ lie on the same face when the graph is drawn on a plane. By appealing to the planar dual of our graph, we conclude that the parametric complexity of the $(s,t)$-cut problem in planar graphs is also $n^{\Omega(\log n)}$.
    
    \paragraph{\textit{Undirected graphs:}} Our construction yields a directed graph, but with a slight modification (by increasing all edge costs uniformly), we obtain an undirected graph with the same number of pieces. Thus our $n^{\Omega(\log n)}$ lower bound holds for \emph{undirected planar graphs} as well.
    
    \paragraph{\textit{Optimization problems:}} Parametric shortest paths have been studied extensively in the optimization literature because of their close connection with several other problems. We briefly mention four.
    
    \begin{itemize}
        \item Nikolova, Kelner, Brand \& Mitzenmacher~\cite{nikolova} consider a stochastic optimization problem on graphs whose edge weights represent random Gaussian variables and where one is required to determine the $s$-$t$ path whose total cost is most likely to be below a specified threshold (the deadline). They provide an $n^{O(\log n)}$ time algorithm for the problem for general graphs, and suggest that when restricted to planar graphs their algorithm might run in polynomial time because the number of extreme points of the \emph{shadow dominant} (a notion closely related to parametric shortest path complexity) is likely to be polynomially (perhaps even linearly) bounded. Our result unfortunately belies this hope.
        \item Correa, Harks, Kreuzen \& Matuschke~\cite{jannik} study the problem of \emph{fare evasion in transit networks}, and consider strategies based on random checks for the service providers, and the response of the users to such strategies. For one of the problems, referred to as the non-adaptive followers' minimization problem, they devise an algorithm based on the parametric shortest path problem, and point out that their algorithm would run in polynomial time on planar graphs if Nikolova's conjecture were to hold.
        
        \item Chakraborty, Fischer, Lachish \& Yuster~\cite{chakraborty} provide two-phase algorithms for the parametric shortest path problem, where the first stage does preprocessing after which an advice is stored in memory so that the algorithm can answer queries efficiently thereafter. A natural application for such an algorithm is traffic networks. Since traffic networks tend to be planar, a good upper bound on the parametric complexity of planar graphs would have allowed for substantial savings in space.
        
        \item We also mention work on a closely related problem. Erickson~\cite{erickson} reformulates an $O(n\log n)$ time algorithm of Borradaile \& Klein~\cite{klein} for max-flows in planar graphs by considering parametric shortest path trees (see Karp \& Orlin~\cite{karporlin}) in the dual graph. He shows that the tree can undergo only a limited number of changes. However, in Erickson's setting, the coefficient of $\lambda$ in the edge weights is always $-1$. He also points out that a similar approach for max-flows in graphs drawn on a torus fails to yield an efficient algorithm because the tree might undergo $\Omega(n^2)$ changes.
    \end{itemize}

    \subsection{Overview of our proof techniques}
    
    We recall two earlier efforts aimed towards resolving Nikolova's conjecture. In her PhD thesis, Nikolova~\cite{nikolovathesis} considers planar embeddings of planar graphs, and shows that the edges can always be assigned weights in such a way that the number of pieces (in the piece-wise linear plot of the shortest paths) is at least the number of faces in the embedding. Note, however, that the number of pieces in the $n$-vertex planar graphs constructed using this approach is at most $2n$. We are aware of only one work that establishes a better upper bound for a family of planar graphs: Correa \emph{et al.}~\cite{jannik} observe that for series parallel graphs, Nikolova's conjecture is true; the parametric complexity of series-parallel graphs is in fact $O(n)$.
    
    \paragraph{\textit{Proof techniques for the main result:}} It is instructive\footnote{As perhaps many others did before us, we initially believed that Nikolova's conjecture was true and tried to prove it.} to briefly review the upper bound arguments of Gusfield and Dean with the hope of tightening them in the setting of planar graphs. Let $G[n,m]$ denote a directed acyclic graph $G$ with vertices $s$ and $t$ that has $m$ layers of $n$ vertices each in between $s$ and $t$. Fix a numbering of the vertices ($1,2,\ldots,n$) in each layer. These arguments are based on the following observations. Let us assume that the shortest $s$-$t$ path is constructed in such a way that starting from $s$ we always move to the neighbour with the shortest distance to $t$, choosing the neighbour having the smallest number when there is a tie. Let $(p_1,p_2,\ldots,p_T)$ be the sequence of shortest paths corresponding to the lower envelope, where each path $p_i$ is constructed in this fashion. This sequence of paths has the following \emph{alternation-free property} (called \emph{expiration property} by Nikolova~\cite{nikolovathesis}). For a path $p$, and vertices $u$ and $v$ that appear on it in that order, let $p[u:v]$ be the subpath of $p$ that connects $u$ to $v$.
    
    \begin{Proposition}[Alternation-free property, expiration property] \label{prop:alt-free}
    Suppose vertices $u$ and $v$ both appear on the three paths $p_i$, $p_j$ and $p_k$ in the sequence $(p_1,p_2,\ldots,p_T)$, where $1\leq i < j < k\leq T$. Furthermore, suppose $q=p_i[u:v]=p_k[u:v]$. Then, $p_j[u:v]=q$.
    \end{Proposition}
    
    This alternation-free property is important because the length of a longest alternation-free sequence of paths in $n$-vertex planar graphs is an upper bound on $\compl^{\pl}(n,\infty)$.
    
    \begin{restatable}[Alternation-free sequences of paths in layered graphs]{Theorem}{thmaltfree}\label{thm:altfree}
        Let $f(n,m)$ be the length of a longest alternation-free sequence paths in the layered graph $G[n,m]$; let $f^{\pl}(n,m)$ be the length of a longest alternation-free sequence of paths in a planar subgraph of $G[n,m]$ (with vertices $s$ and $t$ included). Then, we have the following.
        \[ (i)\,\ f(n,2^k)\geq n^k, (ii)\,\ f^{\pl}(n,(n-1)2^k) \geq n^k. \]
    \end{restatable}
    
    Using the alternation-free property, one observes that $f(n,1)=n$ and $f(n,2^k-1) \leq 2n f(n,2^{k-1}-1)$, which yields $f(n,2^{k}-1) \leq \frac{1}{2}(2n)^k$, implying that $\compl(n,\infty) = n^{O(\log n)}$. Note that the \emph{non-planar graphs} with high parametric shortest path complexity constructed by Carstensen~\cite{carstensenthesis} and Mulumuley \& Shah~\cite{mulmuleyshah} imply that $f(n,n) \geq  n^{\delta  \log n }$ (for some $0<\delta<1$). In~\autoref{sec:words}, we present a construction which shows that $f(n,2^k) \geq n^{k}$. Thus, we have $n^k \leq f(n,2^k) \leq \frac{1}{2} (2n)^k$. More crucially, our construction can be adapted to \emph{planar graphs}.
    
    In~\autoref{sec:binarywords}, we present the construction for planar graphs in detail. This shows that just using the alternation-free property alone, we cannot hope to obtain significantly better upper bounds on $\compl^{\pl}(n,\infty)$. While this construction provides some evidence against Nikolova's conjecture, it does not immediately refute it. In fact, there exist examples of alternation-free sequences of paths in planar graphs that do not arise as parametric shortest paths. Kuchlbauer~\cite[Example 3.11]{martina} presents a planar graph that admits an \emph{infeasible} alternation-free sequence with $10$ paths; that is, no assignment of linear functions to the edges can realize this sequence of $10$ paths as shortest paths.
    
    Our refutation of Nikolova's conjecture is based on the Mulmuley-Shah construction~\cite{mulmuleyshah}. Their construction uses an intricate inductive argument involving the composition of dense bipartite graphs. These bipartite graphs contain large complete bipartite graphs, and are therefore far from planar. We show that, nevertheless, these non-planar bipartite graphs can be simulated by a planar gadget, where each edge is replaced by a path consisting of up to $n^2$ edges and the original weight is carefully distributed among these edges. For this we introduce two ideas. First, staying with the original non-planar construction, we modify the edge weights so that they vary in a structured way. Second, we imagine that the original bipartite graph is drawn on a plane by connecting dots using straight lines, a new vertex arising whenever two straight lines intersect. This results in several new vertices, and spurious paths that do not correspond to any edge of the original bipartite graph. However, the costs of the new edges are so assigned that these spurious paths have much higher costs than the direct path corresponding to the edge in the original bipartite graph. We devote~\autoref{sec:planarlinkgadget} to the construction of this gadget.
    
    The main technique in our construction goes back to Carstensen's work. Our planarization is straightforward in hindsight. The reasons this was not observed before are perhaps the following: (i) the earlier recursive constructions even for general graphs are complicated and not easy to take apart and examine closely (in particular, the Mulmuley-Shah paper is rather cryptic and has errors that throw the reader off); (ii) simple methods of constructing planar graphs with many pieces (in the piece-wise linear plot) tend to navigate around regions in the planar drawing one at a time, somehow (mis)leading one to believe that the limited number of planar regions ought to impose a polynomial upper bound on the number of pieces. See~\autoref{sec:maintheorem} for a detailed proof.

    \paragraph{\textit{Proof techniques for the upper bounds:}} Recall the upper bound arguments for alternation-free paths leading to the recurrence stated after~\autoref{prop:alt-free}. When edge weights vary as degree $d$ polynomials and not just linearly, paths are no longer strictly alternation-free; rather, two paths could alternate up to $d$ times. The complication arising out of this is related to \emph{Davenport-Schinzel sequences}, which have been studied extensively in the discrete geometry literature. The existing upper bounds on the length of Davenport-Schinzel sequences can be combined with the approach of Gusfield and Dean to yield~\autoref{thm:complpoly}. See~\autoref{sec:complpoly} for a detailed proof.
    
    However, when edge weights have the form $a\lambda_1 + b\lambda_2 + c \lambda_3$, our proof techniques depart substantially from the arguments used by Gusfield (which were adapted to univariate polynomials to obtain~\autoref{thm:complpoly}). Although Gusfield's bound is stated for edge weights of the form $a\lambda + b$, essentially the same upper bound holds when the edge weights have the form $a\lambda_1 + b\lambda_2$. (This can be seen, for example, by dividing all costs uniformly by $\lambda_2$.) In the three-parameter setting, it is not clear how one can impose a meaningful linear order on the set $\mathbb{R}^3$, and invoke combinatorial notions such as alternation-free sequences. Instead, we approach the problem geometrically.
    
    Note that the cost of an $s$-$t$ path $P$ has the form $a_P\lambda_1 + b_P\lambda_2 + c_P\lambda_3$, and may be viewed as a vector $(a_P,b_P,c_P)$ in $\mathbb{R}^3$.
    Consider the convex hull of these path vectors. The crucial observation is that for each $\vlambda=(\lambda_1,\lambda_2,\lambda_3)\in\mathbb{R}^3$, at least one of the vertices (vertex here means $0$-dimensional face) of the convex hull corresponds to a minimum cost path. Thus, we need an upper bound on the number of vertices of the convex hull. From here on, our argument is similar to Gusfield's but is carried out in the language of polytopes. The key non-trivial step in the analysis of the number of vertices in the polytope arises when two graphs are placed in series. This is addressed by bounding the number of vertices of the \emph{Minkowski sum} of the polytopes of the constituent graphs. A detailed proof of~\autoref{thm:complmult} is presented in~\autoref{sec:complmult}.
    
    We now briefly point out that there is an obstruction to a simple reduction to the case of two variables. Suppose we set $\lambda_3$ to $1$. Then we are left with a situation where the cost of each $s$-$t$ path is a plane in $\mathbb{R}^3$. The shortest path function $C(G)(\lambda_1,\lambda_2)$ is then a dome-like structure, a polyhedron representing the lower envelope of the planes corresponding to the various $s$-$t$ paths. We wish to bound the number of faces of this polyhedron. Note that every planar cross-section of this polyhedron is a piece-wise linear function in $\mathbb{R}^2$, which by Gusfield's bound has $n^{\log n+O(1)}$ pieces. A conjecture of Shephard~\cite{Shephard} states that the number of faces in the polyhedron is bounded by a polynomial in the maximum number of pieces in such a cross-section. Unfortunately, this conjecture is false: there are polyhedrons with $n$ faces for which the number of pieces in every planar cross-section is at most $O\left(\log n/\log \log n\right)$~\cite{Chazelle,Lagarias}.

    \section{The planar construction with linearly varying edge weights}
    
    In this section, we construct a planar gadget which will be used to construct planar graphs with high shortest path complexity. Our construction closely follows the construction of Mulmuley \& Shah~\cite{mulmuleyshah}, which in turn was based on the construction of Carstensen~\cite{carstensenthesis}. These earlier constructions (and ours) proceed by induction, where we begin with a small base graph, and at each induction step, we increase the number of vertices by a constant factor and the number of pieces in the lower envelope by a factor $n$. After $m$ steps of induction, we obtain a graph with $\poly(n)\cdot\exp(O(m))$ vertices and $n^m$ pieces. \autoref{fig:bighero} illustrates this assembly for $m=3,n=3$, following the template of~\autoref{fig:glgmgr}. We will explain this construction in detail later. The edge weights in the constituent graphs in~\autoref{fig:glgmgr} are carefully chosen, but are not important to our top-level view. The only new component added in each level of induction is the part labelled $\mathsf{LINK}$.
    
    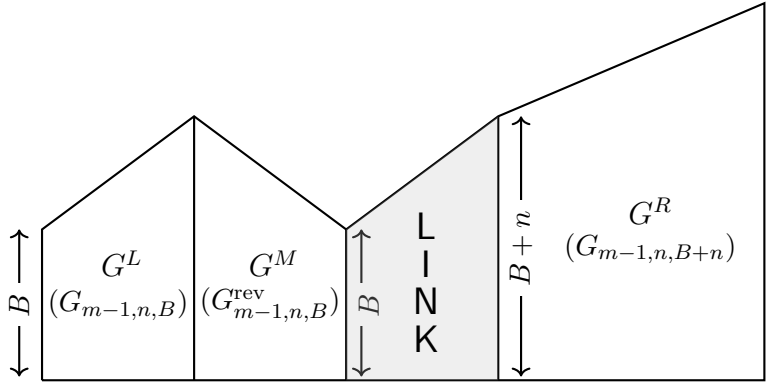
\begin{figure}
    \begin{center}
    \begin{tikzpicture} [scale=0.50, thick,every node/.style={},every fit/.style={ellipse,draw,inner sep=-12pt,text width=2cm}]
    
	\draw [<->] (12.6,7) -- (12.6,0.1);
	\node[fill=white, rotate=90] at (12.6,3.5) {$B+n$};
	\draw [<->] (-0.6,4) -- (-0.6,0.1);
	\node[fill=white, rotate=90] at (-0.6,2) {$B$};
	\draw [<->] (8.4,4) -- (8.4,0.1);
	\node[fill=white, rotate=90] at (8.6,2) {$B$};
    
    \draw (0,0)--(0,4)--(4,7)--(8,4)--(12,7)--(19,10)--(19,0)--(12,0)--(8,0)--(4,0)--(2,0)--(0,0);
    \draw (4,7)--(4,0); \draw (8,4)--(8,0); \draw (12,7)--(12,0);

    \fill[gray!50,nearly transparent] (8,4) -- (12,7) -- (12,0) -- (8,0) -- cycle;
    
    \node[] at (2.1,3.1) {\large{$\GL$}};
    \node[] at (2,2) {$(\Gmn[-1]{B})$};
    
    \node[] at (6.1,3.1) {\large{$\GM$}};
    \node[] at (6,2) {$(\Gmn[-1]{B}^{\rev})$};
    
    \node[] at (10.1,4.1) {\Large{\textsf{L}}};
    \node[] at (10.1,3.1) {\Large{\textsf{I}}};
    \node[] at (10.1,2.1) {\Large{\textsf{N}}};
    \node[] at (10.1,1.1) {\Large{\textsf{K}}};
    
    \node[] at (16,4.5) {\large{$\GR$}};
    \node[] at (16,3.5) {$(\Gmn[-1]{B+n})$};

    \end{tikzpicture}
    \caption{$\Gmn{B}$ is obtained by composing $\Gmn[-1]{B}$, $\Gmn[-1]{B}^{\rev}$, the linking gadget, and $\Gmn[-1]{B+n}$}
    \label{fig:glgmgr}
    \end{center}
    \end{figure}
    
    Our first observation is that the edge weights used by Mulmuley \& Shah in $\mathsf{LINK}$ can be modified so that they have a regular form. Our second observation is that with the modified edge weights, $\mathsf{LINK}$, which is a dense bipartite graph, can be simulated by a planar gadget.
    
    In the following sections, we provide detailed justification for the two contributions outlined above. In~\autoref{sec:maintheorem}, we show that the new edge weights in $\mathsf{LINK}$ also result in a large number of pieces in the lower envelope. In~\autoref{sec:planarlinkgadget}, we show that the non-planar graph $G^{\npl}$ can be simulated by a suitable weighted planar graph; this step, which is at the core of our contribution, has a simple implementation with an appealing proof of correctness.

    \subsection{The linking gadget} \label{sec:planarlinkgadget}

    A linking gadget $\linkgadget(B,n)$ is a bipartite graph $G(U,V,E,(w_e:e \in E))$ with $U=\{0,1,\ldots, B-1\}$, $V=\{0,1,\ldots,B+n-1\}$, $E=\{(b,b+r): b \in U, r =0,1,\ldots,n\}$.  In this graph the cost of the shortest path from vertex $b$ to vertex $j$ is precisely $w_{(b,j)}$ (we often write $w_{b,j}$ instead). We would like to obtain a directed planar simulation of this behaviour.
    
    Let $G^{\pl}$ be the directed graph drawn on a planar strip in $\mathbb{Q}^2$ given by $[0,1]\times[0,2n-2]$; the vertices of $G^{\pl}$ include the sets of points  $\{0\} \times U$ and $\{1\} \times V$; the rest of the graph is obtained as follows. We draw the line segments $\ell_{(b,j)}$ joining $(0,b)$ to $(1,j)$ whenever $(b,j) \in E(G)$, and include all intersection points of such segments in the vertex set of $G^{\pl}$ (see~\autoref{fig:linknpl}). The edge $(u,v)$ is in $G^{\pl}$ if $v$ immediately follows $u$ on some line segment $\ell_e$. The edge weight $w_e$ of the edge $e \in E(G)$ is distributed uniformly among the various edges of $G^{\pl}$ that arise out of $e$. Suppose the vertices $u=(u_x,u_y)$ and $v=(v_x,v_y)$ appear consecutively on $\ell_e$ (note $v_x > u_x, v_y\geq u_y$); then $w_{u,v} = w_e \cdot (v_x - u_x)$.
    
    This completes the description of the weighted planarization $G^{\pl}$ of $G$. The locations of the vertices in this special planar embedding of $G^{\pl}$ are not essential for our construction. However, one feature of this embedding is useful in our proof. A vertex is placed at a point of intersection of two lines of the form $Y = m_1 X + c_1$ and $Y = m_2 X + c_2$; so its $x$-coordinate, namely $(c_2 - c_1)/(m_1-m_2)$ can be written as a fraction with denominator at most $n$, which leads to the following observation.
    
    \begin{Fact} \label{defacto}
    The horizontal distance traversed by an edge $((x_1,y_1),(x_2,y_2))$ of $G^{\pl}$ $($which is $x_2-x_1)$ can be expressed as a non-zero fraction with denominator at most $n^{2}$.
    \end{Fact}
    
    We will use this observation later (\autoref{cl:prob}). Before that, we need to formalize what it means for $G^{\pl}$ to mimic $G$.
    
    \begin{figure}
    \begin{center}
    \begin{tikzpicture} [scale=0.45, thick,every node/.style={},every fit/.style={ellipse,draw,inner sep=-12pt,text width=2cm}]
    
    \def \hstart {2}
    \def \vgap {2.3}
    \def \hgap {9}
    \def \n {3}
    \def \m {3}
    \def \mn {6}
    \def \lshift {21}
    \def \vertcolor {blue}
    \def \intcolor {blue!50}
    \def \sizeintvert {6}
    
    \begin{scope}[very thick,decoration={markings,mark=at position 0.8 with {\arrow{>}}}]
    \foreach \i in {0,...,\n}
    {
        \vertex(b\i) at (\hstart,\i*\vgap) [label=left:$\i$, fill=\vertcolor, \vertcolor] {};
        \foreach \r in {0,...,\m}
            \draw [postaction={decorate},line width=0.1mm] (b\i)--(\hstart+\hgap,\i*\vgap+\r*\vgap);
    }
    \end{scope}
    \foreach \i in {0,...,\mn}
        \vertex(c\i) at (\hstart+\hgap,\i*\vgap) [label=right:$\i$, fill=\vertcolor, \vertcolor] {};
    
    \node at (\hstart*0.5+\lshift*0.5+\hstart*0.5+\hgap*0.5,\n*\vgap-1.2) {\fontsize{30}{36}\selectfont{$\rightsquigarrow$}};
    \node at (\hstart*0.5+\lshift*0.5+\hstart*0.5+\hgap*0.5,\n*\vgap) {\small{planarize}};
    
    \begin{scope}[very thick,decoration={markings,mark=at position 0.8 with {\arrow{>}}}]
    \foreach \i in {0,...,\n}
    {
        \vertex(b\i) at (\hstart+\lshift,\i*\vgap) [label=left:$\i$, fill=\vertcolor, \vertcolor] {};
        \foreach \r in {0,...,\m}
            \draw [line width=0.1 mm] (b\i)--(\hstart+\hgap+\lshift,\i*\vgap+\r*\vgap);
    }
    \end{scope}
    
    \foreach \i in {0,...,\mn}
        \vertex(c\i) at (\hstart+\hgap+\lshift,\i*\vgap) [label=right:$\i$, fill=\vertcolor, \vertcolor] {};
    
    \foreach \i in {1,...,\m}
    {
        \vertex at (\hstart+\lshift+\hgap/3,\i*\vgap) [minimum size=\sizeintvert,fill=\intcolor,\intcolor] {};
        \vertex at (\hstart+\lshift+\hgap/2,\i*\vgap) [minimum size=\sizeintvert,fill=\intcolor,\intcolor] {};
        \ifthenelse{\i=1}{}{\vertex at (\hstart+\lshift+2*\hgap/3,\i*\vgap) [minimum size=\sizeintvert,fill=\intcolor,\intcolor] {};}
        \vertex at (\hstart+\lshift+\hgap/2,\i*\vgap+0.5*\vgap) [minimum size=\sizeintvert,fill=\intcolor,\intcolor] {};
    }
    \end{tikzpicture}
    \caption{The $\mathsf{LINK}$ gadget $\linkgadget(B,n)$ for $B=4, n=3$ and its planarization $\linkgadget(B,n)$}
    \label{fig:linknpl}
    \end{center}
    \end{figure}
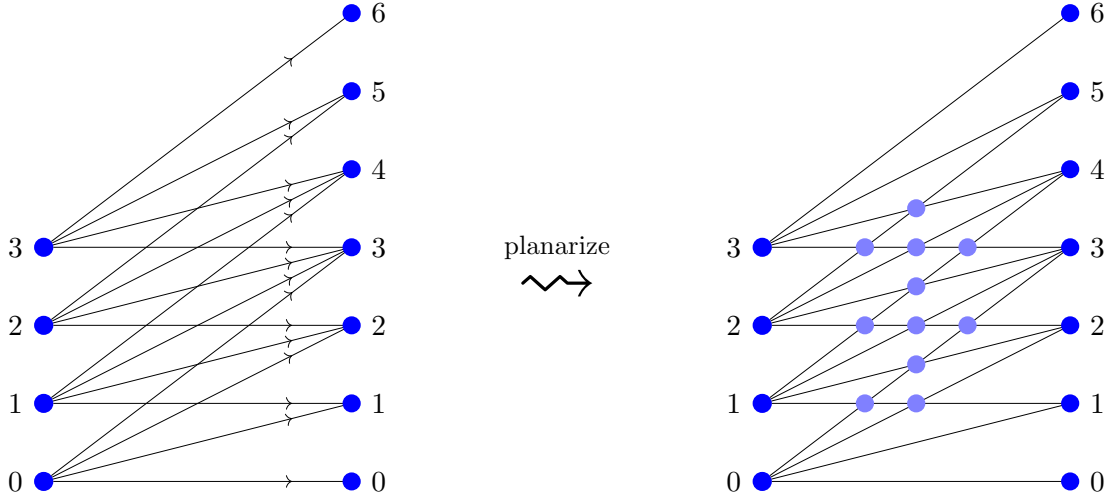
    
    \begin{Definition} We say that $G^{\pl}$ \emph{faithfully simulates $G$} if for all $(b,j) \in U \times V$:
    \begin{enumerate}
        \item[$($i$)$] if $b\leq j \leq b+n$, then the edges arising from the line segment $(0,b)$ to $(1,j)$ form the unique shortest path from $(0,b)$ to $(1,j)$ in $G^{\pl}$,
        \item[$($ii$)$] if $b\leq j \leq b+n$, then the cost of the shortest path from $(0,b)$ to $(1,j)$ in $G^{\pl}$ is precisely $w_{b,j}$, and the cost of every other path from $(0,b)$ to $(1,j)$ is at least $w_{b,j}+1$, and
        \item[$($iii$)$] if $j < b$ or $j > b+n$, then there is no path from $(0,b)$ to $(1,j)$ in $G^{\pl}$.
    \end{enumerate}
    In spirit, this definition says that shortest paths in $G^{\pl}$ should look like edges of $G$.
    \end{Definition}
    
    \begin{Lemma} \label{lm:planarize}
    Suppose $J:E(G) \rightarrow \mathbb{Z}$, and $K$ and $L$ are constants such that $$K \geq n^2\,\left(1+2\underset{e\in E(G)}{\max} |J(e)|\right)\,.$$ Consider a graph $G(U,V,E,w)$ of the form described above with edge weights $$w_{b,b+r} = J(b,b+r) + K\left(\frac{r(r+1)}{2}\right) + L r \lambda,\qquad\text{where } 0\leq b\leq B-1\mbox{ and } 0\leq r\leq n.$$ Then $G^{\pl}$ faithfully simulates $G$. Note that the rate of change of $w_{b,b+r}$ w.r.t. $\lambda$ is $Lr$, where $r$ is the slope of the line segment $\ell_{(b,b+r)}$ w.r.t. the $X$-axis.
    \end{Lemma}
    \begin{proof}
    Consider vertices $b \in U$ and $j \in V$ such that $0\leq b \leq j \leq b+n$.  Consider the path $P$ in $G^{\pl}$ that takes edges along the line segment $\ell_{(b,j)}$. This path has cost $w_{b,j}$. We will show that all other paths from $(0,b)$ to $(1,j)$ have strictly greater cost. Let $r=j-b$ be the \emph{slope} of the line segment $\ell_{(b,j)}$. Suppose $Q$ is another path in $G^{\pl}$ from $(0,b)$ to $(1,j)$.
    
    We make the following claim.
    
    \begin{Claim} \label{cl:prob}
    $\Cost(Q) - \Cost(P) \geq n^{-2} K - 2 \underset{e \in E(G)}{\max} |J(e)|$.
    \end{Claim}
    \begin{proof}[Proof of~\autoref{cl:prob}]
    Let $Q$ consist of vertices $q_0=(x_0,y_0), q_1=(x_1,y_1), q_2=(x_2,y_2), \ldots, q_t=(x_t,y_t)$, where $(x_0,y_0) = (0,b)$ and $(x_t,y_t) = (1,j)$ for some $b,j$. For $i=1,2,\ldots,t$, let $r_i = (y_i-y_{i-1})/(x_i-x_{i-1})$ denote the slope of the edge $(q_{i-1},q_i)$; let $\rho_i = x_i - x_{i-1}$. Then for $i \in \{1,2,\ldots,t\}$, we have
    \begin{align*}
        \rho_i & \geq n^{-2}; \qquad(\text{using~\autoref{defacto}})\\
        r_i & \in \{0,1,\ldots,n\};\\
        r   &= \sum_{i=1}^t \rho_i r_i=j-b;\\
        w_{q_{i-1},q_i} &= \left(J(y,y+r_i) + K \left(\frac{r_i(r_i+1)}{2}\right) + L r_i \lambda\right)(x_{i} - x_{i-1}) \qquad(\text{where } y=y_i-r_ix_i).
    \end{align*}
    Since $0\leq \rho_i\leq 1$ and $\sum_{i=1}^t \rho_i=1$, we may define a random variable $\mathbf{i}$, that takes the value $i \in \{0,1,\ldots,n\}$ with probability $\rho_i$. With this notation, we have $\Expt\left[r_\mathbf{i}\right]=r$.
    
    \paragraph{\textit{Alternative view:}} This paragraph provides a physics perspective to the proof. The calculations can be carried out without reading this paragraph. We may view the path as a height versus time graph of a moving particle, with the $X$-axis representing time and the $Y$-axis representing height. Then, $\Expt\left[r_\mathbf{i}\right] = r = j-b$ corresponds to the fact that the particle underwent a displacement of $r$ units in one unit of time. Note that this claim holds independent of the path taken by the particle. As a result, for the purpose of comparing costs of paths, we can ignore the terms containing $\lambda$. Let us now proceed with the calculations.
    \begin{align}
        \Cost(Q) -\Cost(P) &\geq -2\max_{e} |J(e)| +K\left( \Expt\left[\frac{r_\mathbf{i}^2}{2}\right] - \frac{r^2}{2} \right) + \left(\frac{K+2L\lambda}{2}\right)(\Expt\left[r_\mathbf{i}\right] -r) \\
        &\geq -2\max_{e} |J(e)| +K\left( \Expt\left[\frac{r_\mathbf{i}^2}{2}\right] - \frac{r^2}{2}\right) \qquad\qquad(\text{because } \Expt[r_\mathbf{i}]=r)\\
        &\geq -2\max_{e} |J(e)| + \frac{K}{2} \var[r_\mathbf{i}]. \label{eq:var}
    \end{align}
    We show a lower bound for $\var[r_\mathbf{i}]$. Since $Q$ deviates from $P$, it has at least two edges whose slopes, say $r_{i_1}$ and $r_{i_2}$, differ from $r$ (by at least $1$). Then,
    \[ \var[r_\mathbf{i}] \geq \rho_{i_1}(r_{i_1}-r)^2  + \rho_{i_2}(r_{i_2}-r)^2 \geq 2n^{-2}.\]
    Combining this with~\eqref{eq:var} establishes~\autoref{cl:prob}.
    \end{proof}

    The assumption on $K$ then implies that $P$ is the unique shortest path from $(0,b)$ to $(1,j)$, and the cost of every other path $Q$ from $(0,b)$ to $(1,j)$  is at least $w_{b,j} +1$. This proves (i) and (ii). Finally, (iii) holds because every edge in $G^{\pl}$ corresponds to a line segment with slope at least $0$ and at most $n$. This completes the proof of~\autoref{lm:planarize}.
    \end{proof}
    
    
    \section{Proof of the main result} \label{sec:maintheorem}
    
    \begin{sidewaysfigure}
    \begin{center}
    \begin{tikzpicture} [scale=0.55,thick,transform shape,every node/.style={},every fit/.style={ellipse,draw,inner sep=-12pt,text width=2cm}]
    
    \def \n {3}
    \def \nn {6}
    \def \nnn {9}
    \def \mo {22}
    \def \mi {15}
    \def \mii {8}
    \def \miii {1}
    \def \gap {0.5}
    \def \vertcolor {blue}
    \def \intcolor {blue!50}
    
    \vertex(a001) at (6,\mo) [fill=\vertcolor, \vertcolor, minimum size=4 pt] {};
    \vertex(b001) at (7,\mo) [fill=\vertcolor, \vertcolor, minimum size=4 pt] {};
    \draw (a001)--(b001); 
    \node at (6.5,\mo+1) {\LARGE{$m=0,B=1$}};
    
    \foreach \i in {0,...,\n} 
    {
        \vertex(a10\i) at (15,\gap*\i+\mo) [fill=\vertcolor, \vertcolor, minimum size=4 pt] {};
        \vertex(b10\i) at (16,\gap*\i+\mo) [fill=\vertcolor, \vertcolor, minimum size=4 pt] {};
        \draw (a10\i)--(b10\i);
    }
    \node at (15.5,\gap*\n+\mo+1) {\LARGE{$m=0,B=n+1$}};
    
    \foreach \i in {0,...,\nn} 
    {
        \vertex(a20\i) at (24,\gap*\i+\mo) [fill=\vertcolor, \vertcolor, minimum size=4 pt] {};
        \vertex(b20\i) at (25,\gap*\i+\mo) [fill=\vertcolor, \vertcolor, minimum size=4 pt] {};
        \draw (a20\i)--(b20\i);
    }
    \node at (24.5,\gap*\nn+\mo+1) {\LARGE{$m=0,B=2n+1$}};
    
    \foreach \i in {0,...,\nnn} 
    {
        \vertex(a30\i) at (33,\gap*\i+\mo) [fill=\vertcolor, \vertcolor, minimum size=4 pt] {};
        \vertex(b30\i) at (34,\gap*\i+\mo) [fill=\vertcolor, \vertcolor, minimum size=4 pt] {};
        \draw (a30\i)--(b30\i);
    }
    \node at (37,\gap*\nnn/2+\mo) {\LARGE{$m=0,B=3n+1$}};


    \fill[gray!50] (11,\mi) -- (12,\mi) -- (12,\gap*\n+\mi) -- cycle;
    
    \vertex(a010) at (9,\mi) [fill=\vertcolor, \vertcolor, minimum size=4 pt] {};
    \vertex(b010) at (10,\mi) [fill=\vertcolor, \vertcolor, minimum size=4 pt] {};
    \vertex(c010) at (11,\mi) [fill=\vertcolor, \vertcolor, minimum size=4 pt] {};
    \draw (a010)--(b010)--(c010); 
    
    \foreach \i in {0,...,\n} 
    {
        \vertex(a11\i) at (12,\gap*\i+\mi) [fill=\vertcolor, \vertcolor, minimum size=4 pt] {};
        \vertex(b11\i) at (13,\gap*\i+\mi) [fill=\vertcolor, \vertcolor, minimum size=4 pt] {};
        \draw (a11\i)--(b11\i);
    }
    \node at (7,\gap*\n/2+\mi) {\LARGE{$m=1,B=1$}};
    
    \fill[gray!50] (20,\mi) -- (20,\gap*\n+\mi) -- (21,\gap*\nn+\mi) -- (21,\mi) -- cycle;
    
    \foreach \i in {0,...,\n} 
    {
        \vertex(a21\i) at (18,\gap*\i+\mi) [fill=\vertcolor, \vertcolor, minimum size=4 pt] {};
        \vertex(b21\i) at (19,\gap*\i+\mi) [fill=\vertcolor, \vertcolor, minimum size=4 pt] {};
        \vertex(c21\i) at (20,\gap*\i+\mi) [fill=\vertcolor, \vertcolor, minimum size=4 pt] {};
        \draw (a21\i)--(b21\i)--(c21\i);
    }
    
    \foreach \i in {0,...,\nn} 
    {
        \vertex(a31\i) at (21,\gap*\i+\mi) [fill=\vertcolor, \vertcolor, minimum size=4 pt] {};
        \vertex(b31\i) at (22,\gap*\i+\mi) [fill=\vertcolor, \vertcolor, minimum size=4 pt] {};
        \draw (a31\i)--(b31\i);
    }
    \node at (20,\gap*\nn+\mi+1) {\LARGE{$m=1,B=n+1$}};
    
    \fill[gray!50] (29,\mi) -- (29,\gap*\nn+\mi) -- (30,\gap*\nnn+\mi) -- (30,\mi) -- cycle;
    
    \foreach \i in {0,...,\nn} 
    {
        \vertex(a41\i) at (27,\gap*\i+\mi) [fill=\vertcolor, \vertcolor, minimum size=4 pt] {};
        \vertex(b41\i) at (28,\gap*\i+\mi) [fill=\vertcolor, \vertcolor, minimum size=4 pt] {};
        \vertex(c41\i) at (29,\gap*\i+\mi) [fill=\vertcolor, \vertcolor, minimum size=4 pt] {};
        \draw (a41\i)--(b41\i)--(c41\i);
    }
    
    \foreach \i in {0,...,\nnn} 
    {
        \vertex(a51\i) at (30,\gap*\i+\mi) [fill=\vertcolor, \vertcolor, minimum size=4 pt] {};
        \vertex(b51\i) at (31,\gap*\i+\mi) [fill=\vertcolor, \vertcolor, minimum size=4 pt] {};
        \draw (a51\i)--(b51\i);
    }
    \node at (34,\gap*\nnn/2+\mi) {\LARGE{$m=1,B=2n+1$}};


    \fill[gray!50] (8,\mii) -- (9,\mii) -- (9,\gap*\n+\mii) -- cycle; 
    
    \fill[gray!50] (2,\mii) -- (3,\mii) -- (3,\gap*\n+\mii) -- cycle;
    
    \vertex(a020) at (0,\mii) [fill=\vertcolor, \vertcolor, minimum size=4 pt] {};
    \vertex(b020) at (1,\mii) [fill=\vertcolor, \vertcolor, minimum size=4 pt] {};
    \vertex(c020) at (2,\mii) [fill=\vertcolor, \vertcolor, minimum size=4 pt] {};
    \draw (a020)--(b020)--(c020); 
    
    \foreach \i in {0,...,\n} 
    {
        \vertex(a12\i) at (3,\gap*\i+\mii) [fill=\vertcolor, \vertcolor, minimum size=4 pt] {};
        \vertex(b12\i) at (4,\gap*\i+\mii) [fill=\vertcolor, \vertcolor, minimum size=4 pt] {};
        \draw (a12\i)--(b12\i);
    }
    
    \fill[gray!50] (6,\mii) -- (5,\mii) -- (5,\gap*\n+\mii) -- cycle;
    
    \vertex(a220) at (6,\mii) [fill=\vertcolor, \vertcolor, minimum size=4 pt] {};
    \vertex(b220) at (7,\mii) [fill=\vertcolor, \vertcolor, minimum size=4 pt] {};
    \vertex(c220) at (8,\mii) [fill=\vertcolor, \vertcolor, minimum size=4 pt] {};
    \draw (a220)--(b220)--(c220); 
    
    \foreach \i in {0,...,\n} 
    {
        \vertex(a32\i) at (4,\gap*\i+\mii) [fill=\vertcolor, \vertcolor, minimum size=4 pt] {};
        \vertex(b32\i) at (5,\gap*\i+\mii) [fill=\vertcolor, \vertcolor, minimum size=4 pt] {};
        \draw (a32\i)--(b32\i);
    }
    
    \fill[gray!50] (11,\mii) -- (11,\gap*\n+\mii) -- (12,\gap*\nn+\mii) -- (12,\mii) -- cycle;
    
    \foreach \i in {0,...,\n} 
    {
        \vertex(a42\i) at (9,\gap*\i+\mii) [fill=\vertcolor, \vertcolor, minimum size=4 pt] {};
        \vertex(b42\i) at (10,\gap*\i+\mii) [fill=\vertcolor, \vertcolor, minimum size=4 pt] {};
        \vertex(c42\i) at (11,\gap*\i+\mii) [fill=\vertcolor, \vertcolor, minimum size=4 pt] {};
        \draw (a42\i)--(b42\i)--(c42\i);
    }    
    
    \foreach \i in {0,...,\nn} 
    {
        \vertex(a52\i) at (12,\gap*\i+\mii) [fill=\vertcolor, \vertcolor, minimum size=4 pt] {};
        \vertex(b52\i) at (13,\gap*\i+\mii) [fill=\vertcolor, \vertcolor, minimum size=4 pt] {};
        \draw (a52\i)--(b52\i);
    }
    \node at (7,\gap*\nn+\mii) {\LARGE{$m=2,B=1$}};
    
    \fill[gray!50] (35,\mii) -- (35,\gap*\n+\mii) -- (36,\gap*\nn+\mii) -- (36,\mii) -- cycle;
    
    \fill[gray!50] (29,\mii) -- (29,\gap*\n+\mii) -- (30,\gap*\nn+\mii) -- (30,\mii) -- cycle;
    
    \fill[gray!50] (33,\mii) -- (33,\gap*\n+\mii) -- (32,\gap*\nn+\mii) -- (32,\mii) -- cycle;
    
    \foreach \i in {0,...,\n} 
    {
        \vertex(a62\i) at (27,\gap*\i+\mii) [fill=\vertcolor, \vertcolor, minimum size=4 pt] {};
        \vertex(b62\i) at (28,\gap*\i+\mii) [fill=\vertcolor, \vertcolor, minimum size=4 pt] {};
        \vertex(c62\i) at (29,\gap*\i+\mii) [fill=\vertcolor, \vertcolor, minimum size=4 pt] {};
        \draw (a62\i)--(b62\i)--(c62\i);
    }    
    
    \foreach \i in {0,...,\nn} 
    {
        \vertex(a72\i) at (30,\gap*\i+\mii) [fill=\vertcolor, \vertcolor, minimum size=4 pt] {};
        \vertex(b72\i) at (31,\gap*\i+\mii) [fill=\vertcolor, \vertcolor, minimum size=4 pt] {};
        \draw (a72\i)--(b72\i);
    }
    
    \foreach \i in {0,...,\n} 
    {
        \vertex(a82\i) at (33,\gap*\i+\mii) [fill=\vertcolor, \vertcolor, minimum size=4 pt] {};
        \vertex(b82\i) at (34,\gap*\i+\mii) [fill=\vertcolor, \vertcolor, minimum size=4 pt] {};
        \vertex(c82\i) at (35,\gap*\i+\mii) [fill=\vertcolor, \vertcolor, minimum size=4 pt] {};
        \draw (a82\i)--(b82\i)--(c82\i);
    }    
    
    \foreach \i in {0,...,\nn} 
    {
        \vertex(a92\i) at (31,\gap*\i+\mii) [fill=\vertcolor, \vertcolor, minimum size=4 pt] {};
        \vertex(b92\i) at (32,\gap*\i+\mii) [fill=\vertcolor, \vertcolor, minimum size=4 pt] {};
        \draw (a92\i)--(b92\i);
    }
    
    \fill[gray!50] (38,\mii) -- (38,\gap*\nn+\mii) -- (39,\gap*\nnn+\mii) -- (39,\mii) -- cycle;
    
    \foreach \i in {0,...,\nn} 
    {
        \vertex(ax2\i) at (36,\gap*\i+\mii) [fill=\vertcolor, \vertcolor, minimum size=4 pt] {};
        \vertex(bx2\i) at (37,\gap*\i+\mii) [fill=\vertcolor, \vertcolor, minimum size=4 pt] {};
        \vertex(cx2\i) at (38,\gap*\i+\mii) [fill=\vertcolor, \vertcolor, minimum size=4 pt] {};
        \draw (ax2\i)--(bx2\i)--(cx2\i);
    }
    
    \foreach \i in {0,...,\nnn} 
    {
        \vertex(ay2\i) at (39,\gap*\i+\mii) [fill=\vertcolor, \vertcolor, minimum size=4 pt] {};
        \vertex(by2\i) at (40,\gap*\i+\mii) [fill=\vertcolor, \vertcolor, minimum size=4 pt] {};
        \draw (ay2\i)--(by2\i);
    }
    \node at (34,\gap*\nnn+\mii) {\LARGE{$m=2,B=n+1$}};


    \fill[gray!50] (26,\miii) -- (27,\miii) -- (27,\gap*\n+\miii) -- cycle; 
    
    \fill[gray!50] (8,\miii) -- (9,\miii) -- (9,\gap*\n+\miii) -- cycle; 
    
    \fill[gray!50] (2,\miii) -- (3,\miii) -- (3,\gap*\n+\miii) -- cycle;
    
    \fill[gray!50] (18,\miii) -- (17,\miii) -- (17,\gap*\n+\miii) -- cycle; 
    
    \fill[gray!50] (24,\miii) -- (23,\miii) -- (23,\gap*\n+\miii) -- cycle;
    
    \vertex(a030) at (0,\miii) [fill=\vertcolor, \vertcolor, minimum size=4 pt] {};
    \vertex(b030) at (1,\miii) [fill=\vertcolor, \vertcolor, minimum size=4 pt] {};
    \vertex(c030) at (2,\miii) [fill=\vertcolor, \vertcolor, minimum size=4 pt] {};
    \draw (a030)--(b030)--(c030); 
    
    \foreach \i in {0,...,\n} 
    {
        \vertex(a13\i) at (3,\gap*\i+\miii) [fill=\vertcolor, \vertcolor, minimum size=4 pt] {};
        \vertex(b13\i) at (4,\gap*\i+\miii) [fill=\vertcolor, \vertcolor, minimum size=4 pt] {};
        \draw (a13\i)--(b13\i);
    }
    
    \fill[gray!50] (6,\miii) -- (5,\miii) -- (5,\gap*\n+\miii) -- cycle;
    
    \vertex(a230) at (6,\miii) [fill=\vertcolor, \vertcolor, minimum size=4 pt] {};
    \vertex(b230) at (7,\miii) [fill=\vertcolor, \vertcolor, minimum size=4 pt] {};
    \vertex(c230) at (8,\miii) [fill=\vertcolor, \vertcolor, minimum size=4 pt] {};
    \draw (a230)--(b230)--(c230); 
    
    \foreach \i in {0,...,\n} 
    {
        \vertex(a33\i) at (4,\gap*\i+\miii) [fill=\vertcolor, \vertcolor, minimum size=4 pt] {};
        \vertex(b33\i) at (5,\gap*\i+\miii) [fill=\vertcolor, \vertcolor, minimum size=4 pt] {};
        \draw (a33\i)--(b33\i);
    }
    
    \fill[gray!50] (11,\miii) -- (11,\gap*\n+\miii) -- (12,\gap*\nn+\miii) -- (12,\miii) -- cycle;
    
    \foreach \i in {0,...,\n} 
    {
        \vertex(a43\i) at (9,\gap*\i+\miii) [fill=\vertcolor, \vertcolor, minimum size=4 pt] {};
        \vertex(b43\i) at (10,\gap*\i+\miii) [fill=\vertcolor, \vertcolor, minimum size=4 pt] {};
        \vertex(c43\i) at (11,\gap*\i+\miii) [fill=\vertcolor, \vertcolor, minimum size=4 pt] {};
        \draw (a43\i)--(b43\i)--(c43\i);
    }    
    
    \foreach \i in {0,...,\nn} 
    {
        \vertex(a53\i) at (12,\gap*\i+\miii) [fill=\vertcolor, \vertcolor, minimum size=4 pt] {};
        \vertex(b53\i) at (13,\gap*\i+\miii) [fill=\vertcolor, \vertcolor, minimum size=4 pt] {};
        \draw (a53\i)--(b53\i);
    }
    
    \vertex(aa30) at (26,\miii) [fill=\vertcolor, \vertcolor, minimum size=4 pt] {};
    \vertex(ba30) at (25,\miii) [fill=\vertcolor, \vertcolor, minimum size=4 pt] {};
    \vertex(ca30) at (24,\miii) [fill=\vertcolor, \vertcolor, minimum size=4 pt] {};
    \draw (aa30)--(ba30)--(ca30); 
    
    \foreach \i in {0,...,\n} 
    {
        \vertex(ab3\i) at (23,\gap*\i+\miii) [fill=\vertcolor, \vertcolor, minimum size=4 pt] {};
        \vertex(bb3\i) at (22,\gap*\i+\miii) [fill=\vertcolor, \vertcolor, minimum size=4 pt] {};
        \draw (ab3\i)--(bb3\i);
    }
    
    \fill[gray!50] (20,\miii) -- (21,\miii) -- (21,\gap*\n+\miii) -- cycle;
    
    \vertex(ac30) at (20,\miii) [fill=\vertcolor, \vertcolor, minimum size=4 pt] {};
    \vertex(bc30) at (19,\miii) [fill=\vertcolor, \vertcolor, minimum size=4 pt] {};
    \vertex(cc30) at (18,\miii) [fill=\vertcolor, \vertcolor, minimum size=4 pt] {};
    \draw (ac30)--(bc30)--(cc30); 
    
    \foreach \i in {0,...,\n} 
    {
        \vertex(ad3\i) at (22,\gap*\i+\miii) [fill=\vertcolor, \vertcolor, minimum size=4 pt] {};
        \vertex(bd3\i) at (21,\gap*\i+\miii) [fill=\vertcolor, \vertcolor, minimum size=4 pt] {};
        \draw (ad3\i)--(bd3\i);
    }
    
    \fill[gray!50] (15,\miii) -- (15,\gap*\n+\miii) -- (14,\gap*\nn+\miii) -- (14,\miii) -- cycle;
    
    \foreach \i in {0,...,\n} 
    {
        \vertex(ae3\i) at (17,\gap*\i+\miii) [fill=\vertcolor, \vertcolor, minimum size=4 pt] {};
        \vertex(be3\i) at (16,\gap*\i+\miii) [fill=\vertcolor, \vertcolor, minimum size=4 pt] {};
        \vertex(ce3\i) at (15,\gap*\i+\miii) [fill=\vertcolor, \vertcolor, minimum size=4 pt] {};
        \draw (ae3\i)--(be3\i)--(ce3\i);
    }    
    
    \foreach \i in {0,...,\nn} 
    {
        \vertex(af3\i) at (14,\gap*\i+\miii) [fill=\vertcolor, \vertcolor, minimum size=4 pt] {};
        \vertex(bf3\i) at (13,\gap*\i+\miii) [fill=\vertcolor, \vertcolor, minimum size=4 pt] {};
        \draw (af3\i)--(bf3\i);
    }
    
    \fill[gray!50] (35,\miii) -- (35,\gap*\n+\miii) -- (36,\gap*\nn+\miii) -- (36,\miii) -- cycle;
    
    \fill[gray!50] (29,\miii) -- (29,\gap*\n+\miii) -- (30,\gap*\nn+\miii) -- (30,\miii) -- cycle;
    
    \fill[gray!50] (33,\miii) -- (33,\gap*\n+\miii) -- (32,\gap*\nn+\miii) -- (32,\miii) -- cycle;
    
    \foreach \i in {0,...,\n} 
    {
        \vertex(a63\i) at (27,\gap*\i+\miii) [fill=\vertcolor, \vertcolor, minimum size=4 pt] {};
        \vertex(b63\i) at (28,\gap*\i+\miii) [fill=\vertcolor, \vertcolor, minimum size=4 pt] {};
        \vertex(c63\i) at (29,\gap*\i+\miii) [fill=\vertcolor, \vertcolor, minimum size=4 pt] {};
        \draw (a63\i)--(b63\i)--(c63\i);
    }    
    
    \foreach \i in {0,...,\nn} 
    {
        \vertex(a73\i) at (30,\gap*\i+\miii) [fill=\vertcolor, \vertcolor, minimum size=4 pt] {};
        \vertex(b73\i) at (31,\gap*\i+\miii) [fill=\vertcolor, \vertcolor, minimum size=4 pt] {};
        \draw (a73\i)--(b73\i);
    }
    
    \foreach \i in {0,...,\n} 
    {
        \vertex(a83\i) at (33,\gap*\i+\miii) [fill=\vertcolor, \vertcolor, minimum size=4 pt] {};
        \vertex(b83\i) at (34,\gap*\i+\miii) [fill=\vertcolor, \vertcolor, minimum size=4 pt] {};
        \vertex(c83\i) at (35,\gap*\i+\miii) [fill=\vertcolor, \vertcolor, minimum size=4 pt] {};
        \draw (a83\i)--(b83\i)--(c83\i);
    }    
    
    \foreach \i in {0,...,\nn} 
    {
        \vertex(a93\i) at (31,\gap*\i+\miii) [fill=\vertcolor, \vertcolor, minimum size=4 pt] {};
        \vertex(b93\i) at (32,\gap*\i+\miii) [fill=\vertcolor, \vertcolor, minimum size=4 pt] {};
        \draw (a93\i)--(b93\i);
    }
    
    \fill[gray!50] (38,\miii) -- (38,\gap*\nn+\miii) -- (39,\gap*\nnn+\miii) -- (39,\miii) -- cycle;
    
    \foreach \i in {0,...,\nn} 
    {
        \vertex(ax3\i) at (36,\gap*\i+\miii) [fill=\vertcolor, \vertcolor, minimum size=4 pt] {};
        \vertex(bx3\i) at (37,\gap*\i+\miii) [fill=\vertcolor, \vertcolor, minimum size=4 pt] {};
        \vertex(cx3\i) at (38,\gap*\i+\miii) [fill=\vertcolor, \vertcolor, minimum size=4 pt] {};
        \draw (ax3\i)--(bx3\i)--(cx3\i);
    }
    
    \foreach \i in {0,...,\nnn} 
    {
        \vertex(ay3\i) at (39,\gap*\i+\miii) [fill=\vertcolor, \vertcolor, minimum size=4 pt] {};
        \vertex(by3\i) at (40,\gap*\i+\miii) [fill=\vertcolor, \vertcolor, minimum size=4 pt] {};
        \draw (ay3\i)--(by3\i);
    }
    \node at (20,\gap*\nnn+\miii-1) {\LARGE{$m=3,B=1$}};
    
    \end{tikzpicture} \caption{An instantiation of the graph $\Gmn{B}$ for $m=3,n=3,B=1$. All the edges shown are of weight $0$ and the gray area is the linking gadget, $\mathsf{LINK}$.} \label{fig:bighero}
    \end{center}
    \end{sidewaysfigure}
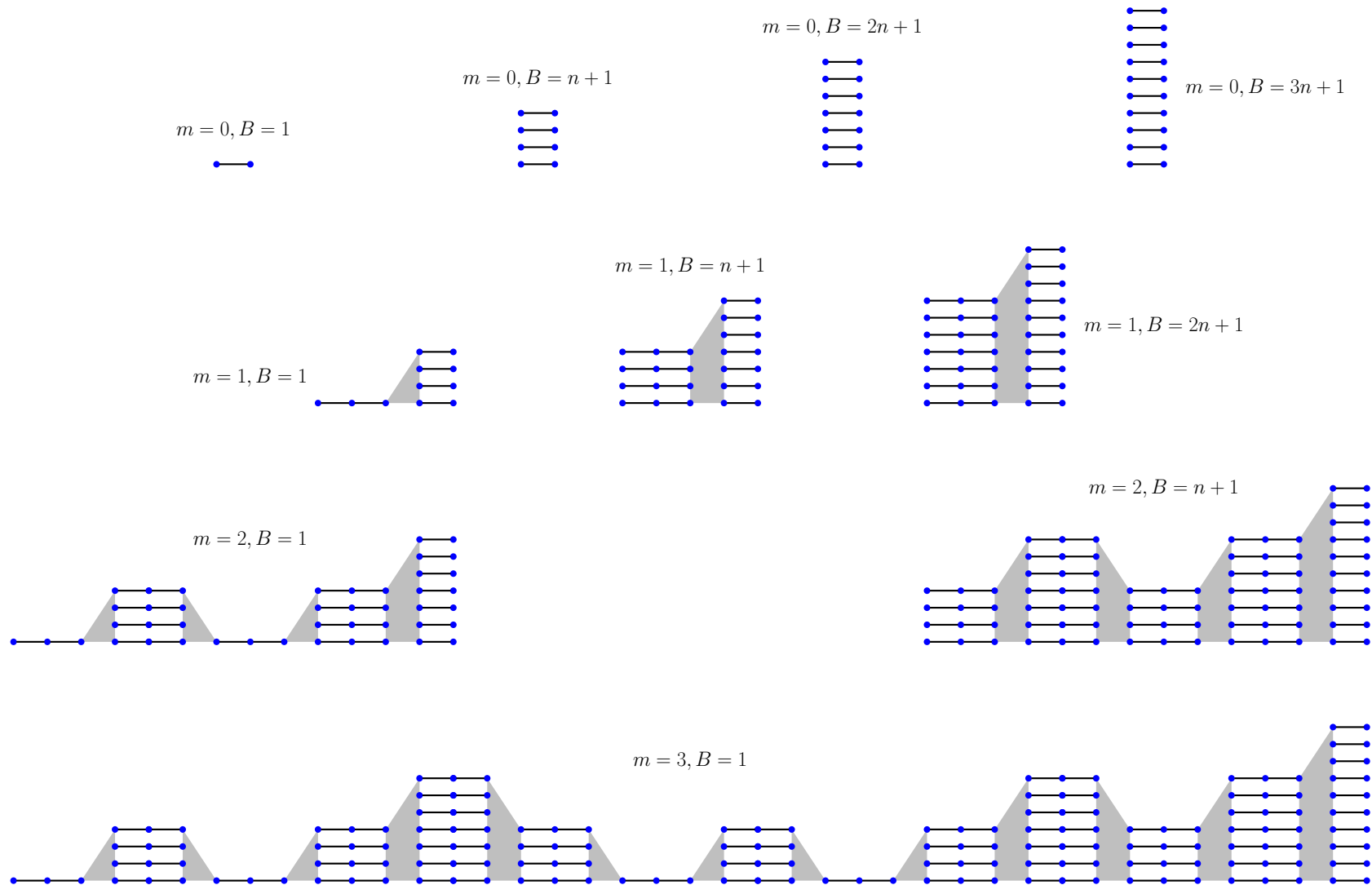
    
    We now prove our main theorem.
    
    \thmmainresult*
    
    Our proof is constructive. Throughout this section, let $n \in \mathbb{N}$ be a fixed natural number and define $N=n^2$. Also, for simplicity, we assume that $n\geq 4$.
    
    The proof is by induction. We begin with a small base graph with one interval which has its own dedicated path (that is, whenever the value of $\lambda$ lies in that interval, its dedicated path is the shortest path in the graph). In each inductive step, we roughly triple the size of the graph and subdivide each interval into $n$ intervals. Each of those intervals have their own dedicated paths. After $\log n$ steps, we end up with a graph on $\poly(n)$ vertices with $n^{\log n}$ intervals.
    
    
    \subsection{Inductive definition of intervals} \label{sec:inddefint}
    In our recursive construction, we will construct paths that reign as the shortest path in particular intervals (as mentioned above, each interval has its dedicated path) for the parameter $\lambda$. We will construct a large number of intervals and show that a different path is the unique shortest path in each interval. This will establish that there are many pieces in the piece-wise linear cost of the shortest path in our graph. The construction of the graph and the role of the intervals is described in detail later. We first place the framework by describing the intervals inductively. The intervals depend on the parameter $N$ (recall $N=n^2$). Then, for $m=0,1,\ldots,\floor{\log n}$, and $j=0,1,\ldots,n^m-1$, we define points $\alpha(j,m) \in \mathbb{Q}$ inductively; these points will be used to define the intervals. Each interval is of length $N-2$. At the $m$-th step of our construction, we have $n^m$ intervals.
    
    \paragraph{Base case ($m=0$):} We set $\alpha(0,0) = 0$. (Since $0\leq j\leq n^m-1$, the only possible value for $j$ is $0$.)
    \paragraph{Induction ($m\geq 1$):}
    For $m\geq 1$ and $0\leq j\leq n^m-1$, we write $j= nd + r$, where $0\leq d \leq n^{m-1}-1$ and $0\leq r \leq n-1$; then, we set $\alpha(j,m)= N\alpha(d,m-1) + N(r+1)$. 
    
    \paragraph{Intervals:} For $m\geq 0$ and $0 \leq j \leq n^m-1$, let $\interval(j,m)=[\alpha(j,m)+1, \alpha(j,m)+N-1]$.
    
    The idea behind this definition of the intervals is as follows. Given a set of $n^{m-1}$ intervals at level $m-1$, we first \emph{stretch} them by a factor $N$ (the corresponding graph can also be appropriately stretched by a factor $N$ by replacing the $\lambda$ in each edge weight with $\lambda/N$; this will be explained in detail later). Then, we subdivide each interval into $n$ disjoint intervals to obtain $n^m$ intervals at level $m$.
    
    However, in order to apply the induction hypothesis cleanly, we require that all of the $n$ subdivided intervals are contained in the stretched version of their parent interval. Our definition of $\alpha(j,m)$ has $N(r+1)$ instead of $Nr$ precisely to ensure that $\interval(j,m)\subseteq N\cdot\interval(d,m-1)$ (the definition in Mulmuley \& Shah~\cite{mulmuleyshah} unfortunately overlooks this point). Let us now summarize these observations.
    
    \begin{enumerate}
        \item For each $m\geq 0$ and $0 \leq j \leq n^m-1$, $|\interval(j,m)|=N-2$.
        \item For each $m\geq 1$ and $0 \leq j \leq n^m-1$, $\interval(j,m)\subseteq N\cdot\interval(d,m-1)$, where $d=\floor{j/n}$.
        \item For each $m\geq 0$, the intervals in the list $(\interval(j,m): j=0,1,\ldots,n^m-1)$ are disjoint.
        \item For each $m\geq 0$ and $0 \leq j \leq n^m-1$, $\interval(j,m)\subseteq[0,N^{m+1}]$. In particular, $\alpha(j,m)\leq N^{m+1}$.
    \end{enumerate}
    
    \subsection{Inductive construction of graphs}
    Our induction depends on three parameters $B$, $D$ and $m$, which impose constraints on the layered, weighted, planar graphs we construct.
    
    \paragraph{The parameter $B$:} $B\in\mathbb{N}$ denotes the number of vertices in the first (input) layer of this graph. $B$ takes values of the form $1,n+1,2n+1,\ldots$, and for a fixed $B$, the variable $b\in\{0,1,\ldots,B-1\}$ denotes an input vertex. All our paths originate in the first layer of the graph and end in the last layer. When we derive our main theorem from this construction, we set $B=1$, which means our final graph has \emph{one} input vertex. We call this unique input vertex $s$ and connect all vertices in the last layer to a new vertex $t$ using edges of weight $0$, so that we have pristine $s$-$t$ paths as promised. Thus, we do not mention $s$ and $t$ for the rest of our proof.
    
    \paragraph{The parameter $D$:} $D\in\mathbb{Q}$, $|D|\leq 1$ is used to determine the weights of the edges.
    
    \paragraph{The parameter $m$:} $m\in\mathbb{N}$, $m\geq 0$ is the induction parameter (this is the same $m$ which is used to define the intervals). 
    
    See~\autoref{fig:bighero} for a  step-by-step visualization of this construction. The formal induction will be carried out using a predicate $\Phi$, which we now define.
    
    \begin{mdframed}[backgroundcolor=gray!10]
    {
	\setlength{\parindent}{0cm}
	\setlength{\parskip}{10pt plus 1pt minus 1pt}
	\begin{center}
        \textbf{\underline{The Predicate $\Phi$}}
    \end{center}
    For $B$, $D$, $m$ as described above, we say that \emph{the predicate $\Phi(B,D,m)$ holds} if there is a layered, weighted, planar graph $G(B,D,m)$ with $B$ input vertices, rational edge weights, and $Bn^m$ paths $P_{bj}$ (for $b=0,1,\ldots, B-1$ and $j=0,1,\ldots,n^m-1$) satisfying the following properties.
    \begin{enumerate}
        \item[$($i$)$] The graph $G(B,D,m)$ has at most $(3^{m+1}-1)(B+ m n)^4$ vertices.
        \item[$($ii$)$] The weight of each edge $e$ in the graph $G(B,D,m)$ has the form $a_e+b_e\lambda$ such that 
        \begin{align*}
            &a_e=a_{1e}+a_{2e}D &\mbox{and} &&b_e=b_{1e}+b_{2e}D,
        \end{align*}
        where $a_{1e}, b_{1e}, a_{2e}, b_{2e}$ are rational numbers with denominator at most $2^mn^2$ and numerator at most $(400\,NB)^{5m^2}$, in absolute value.
        \item[$($iii$)$] For all $b$, $j$, and $\lambda \in \interval(j,m)$, the \emph{unique} shortest path from the input vertex $b$ to the last layer of $G$ is $P_{bj}$, and $\Cost(Q_{b})(\lambda) - \Cost(P_{bj})(\lambda) \geq 1$, for all other paths $Q_{b}$ from the input vertex $b$ to the last layer of $G$.
        \item[$($iv$)$] For all $b$ and $j$, \[\Cost(P_{bj})(\lambda) = \Cost(P_{0j})(\lambda) + bD\alpha(j,m).\]
        \item[$($v$)$] For all $j$, the paths in the list $(P_{bj}:b=0,1,\ldots,B-1)$ are vertex-disjoint.
        \item[$($vi$)$] For all $b$, the paths in the list $(P_{bj}: j=0,1,\ldots,n^m-1)$ are distinct.
    \end{enumerate}
    }
    \end{mdframed}
    
    The following lemma is essentially the same as Lemma 4.1 of Mulmuley \& Shah~\cite{mulmuleyshah}. We closely follow their argument, slightly simplifying the induction, providing more detailed calculations, and correcting some errors; we crucially employ the planarized linking gadget of~\autoref{sec:planarlinkgadget} and~\autoref{lm:planarize} to ensure that our graphs are planar.
    \begin{Lemma} [Main lemma] \label{lm:mainlemma}
    For all integers $B\geq 1$, rational numbers $D\in[-1,+1]$ and integers $m\geq 0$, the predicate $\Phi(B,D,m)$ holds.
    \end{Lemma}
    We will prove this lemma after using it to establish our main theorem.
    \begin{proof}[Proof of~\autoref{thm:mainresult}] Taking $B=1$, $D=0$ and $m=\floor{\log n}$, we conclude that $\Phi(1,0,\floor{\log n})$ holds (\autoref{lm:mainlemma}). Using property (i), the number of vertices in the corresponding graph $G(1,0,\floor{\log n})$ is at most
    \begin{align*}
        (3^{m+1} -1)(B+mn)^4 & \leq (3^{\log n+1}-1)(1+(\log n)n)^4\\
        &\leq (3n^{1.585})(1+(\log n)n)^4 &\text{(since $\log_2 3\leq 1.585$)}\\
        &\leq 6n^{1.585}(n\log n)^4 &\text{(since $n\geq 4$)}\\
        &\leq 6n^{1.585}(n\cdot n^{0.6})^4\\
        &\leq 6n^8.
    \end{align*}
    To this graph we attach a sink vertex $t$ as stated earlier. The graph admits $n^m$ disjoint intervals, with a different unique shortest $s$-$t$ path in each (properties (iii),(vi)); so the cost of the shortest $s$-$t$ path in this graph has $n^{\floor{\log n}}$ pieces in the lower envelope.
    
    Using property (ii) and substituting $B=1$, $D=0$ and $m=\floor{\log n}$, the value of the largest coefficient (numerator or denominator) in the edge weights of the graph is at most
    \begin{align*}
        (400\,NB)^{5m^2}&\leq(400\,n^2)^{5(\log n)^2}\\
        &\leq(400\cdot 2^{2\log n})^{5(\log n)^2}\\
        &\leq(2^{5\log n}\cdot 2^{2\log n})^{5(\log n)^2 } &\text{(since $n\geq 4$)}\\
        &\leq 2^{35(\log n)^3}.
    \end{align*}
    This implies that the bit lengths of the coefficients in the edge weights are bounded by $35(\log n)^3$.
    
    Let $\nu$ be a large positive integer. Let $n$ be the largest integer such that $6n^8+1 \leq \nu$. Note $n=\nu^{\Theta(1)}$ and hence $\log n = \Theta(\log \nu)$. Using the construction above (adding dummy isolated vertices if necessary), we obtain a graph on $\nu$ vertices, whose edge weights have rational coefficients with numerator and denominator of bit lengths bounded by $O((\log \nu)^3)$, and in which the cost of the shortest path has $\nu^{\Omega(\log \nu)}$ pieces in the lower envelope. This completes the proof of our main theorem.
    
    \paragraph{Remark:} Since we require integer edge weights, we can consider clearing all denominators in the coefficients. However, the LCM of the denominators may be prohibitively large. To keep the numbers small, we can modify our construction slightly. The points in the planar linking gadgets can be located at nearby points whose coordinates are multiples of (say) $n^{-4}$. This will ensure that the LCM of the denominators can be written in $O(\log \nu)$ bits. Clearing the denominators now keeps the final integer coefficients $O((\log \nu)^3)$ bits long. 
    
    Note that we did not use property (iv) or (v), which are needed merely in the inductive proof of the main lemma. 
    \end{proof}
    \subsection{Proof of the main lemma}
    
    \begin{proof}[Proof of~\autoref{lm:mainlemma}] We will use induction on $m$ to verify $\Phi(B,D,m)$. For the base case ($m=0$), let $G$ consist of $B$ disjoint edges, each of weight $0$, leaving the $B$ input vertices. The only choice for $j$ in this case is $j=0$ (since $j$ varies from $0$ to $n^m-1$). To verify property (ii), note that each edge weight is of the form $((0+0\cdot D)+(0+0\cdot D)\lambda)$. To verify property (iv), recall that $\alpha(0,0)=0$. All the other properties for $\Phi$ are also easily verified.
    
    Let $m\geq 1$. Assume that $\Phi(B',D',m-1)$ holds for all $B'$ and $D'$. We now fix $B$ and $D$ and show that $\Phi(B,D,m)$ holds for the graph $G(B,D,m)$. Based on $B$, $D$ and $m$, we fix constants
    \begin{align}
    K_L &= 400\,N^{m+4}B^2; \label{eq:setKL}\\
    K_R &= 20\,N^3B; \label{eq:setKR}\\
    D_L &= \frac{N}{2K_L}\left(D - \frac{K_R}{N}\right); \label{eq:setDL}\\
    D_R &= 1. \label{eq:setDR}
    \end{align}
    These constants, which may seem mysterious, will be justified by the claims that follow. Let us now explain the construction and edge weights of $G$.
    
    \paragraph{Construction of $G$:} The graph $G$ is built by concatenating four components: $\GL$, $\GM$, $\mathsf{LINK}$ and $\GR$. Let $\GL$ be the graph corresponding to the induction hypothesis $\Phi(B,D_L,m-1)$; we refer to the corresponding $Bn^{m-1}$ paths by $P_{bd}^L$ where $0\leq b\leq B-1$ and $0\leq d\leq n^{m-1}-1$. Let $\GM$ denote the graph obtained by mirroring $\GL$ about its last layer and reversing the directions of its edges so that all edges go from left to right (see~\autoref{fig:glgmgr}). Thus, $\GM$ has $B$ vertices in its last layer. Let $\GR$ be the graph corresponding to the induction hypothesis $\Phi(B+n,D_R,m-1)$; we refer to the corresponding $(B+n)n^{m-1}$ paths by $P_{bd}^R$ where $0\leq b\leq B+n-1$ and $0\leq d \leq n^{m-1}-1$.
    
    \paragraph{Edge weights of $G$:} We need to transform the edges weights in $\GL$, $\GM$ and $\GR$ before we put them together with a linking gadget to obtain our graph $G$. Let $w_e^L$, $w_e^M$ and $w_e^R$ denote the weights of the edges in $\GL$, $\GM$ and $\GR$, and let $w_e$ denote their weights in $G$.
    \begin{align*}
        w_e(\lambda)&\leftarrow K_L\cdot w_e^L(\lambda/N) &&\forall\, e\in E(\GL)\\
        w_e(\lambda)&\leftarrow K_L\cdot w_e^M(\lambda/N) &&\forall\, e\in E(\GM)\\
        w_e(\lambda)&\leftarrow K_R\cdot w_e^L(\lambda/N) &&\forall\, e\in E(\GR)
    \end{align*}
    In essence, we are \emph{scaling} (by factors $K_L$ and $K_R$)  and \emph{stretching} (by a factor $N$)\footnote{Recall that our intervals are stretched by a factor $N$ when we go from one level of recursion to the next (\autoref{sec:inddefint}).} our already existing solutions for $\GL$, $\GM$ and $\GR$ so that together they can form a solution for $G$.
    
    \paragraph{Linking gadget:} Let $\linkgadget(B,n)$ be the non-planar linking gadget with edge weights
    \begin{equation} \label{eq:linkedges}
        w_{b,b+r} = NDrb + \frac{K_R}{N}\left(\left(\frac{r(r+1)}{2}\right)N - r\lambda\right),\qquad\text{where } 0\leq b\leq B-1 \text{ and } 0\leq r\leq n,
    \end{equation}
    and let $\linkgadget^{\pl}(B,n)$ be its planarized version. Note that the $D$ used in~\eqref{eq:linkedges} is the $D$ that was part of the predicate $\Phi(B,D,m)$ (neither $D_L$ nor $D_R$). The graph $G$ obtained by composing $\GL$, $\GM$, $\linkgadget^{\pl}$ and $\GR$ is shown in~\autoref{fig:glgmgr}. Since $\GL$, $\GM$ and $\GR$ are planar by induction, and $\linkgadget^{\pl}(B,n)$ is planar, the graph obtained by composing them is also planar (\autoref{fig:bighero}). This completes the description of all the constituent components of $G$.
    
    Before we proceed further, let us verify that for our choice of parameters, $\linkgadget^{\pl}$ faithfully simulates its non-planar counterpart. Invoke~\autoref{lm:planarize} with $J(b,b+r)=NDrb$, $K=K_R$ and $L=-K_R/N$. For the setting of $K_R$ in~\eqref{eq:setKR}, we have (recall $N=n^2$)
    \begin{align*}
        n^2\,\left(1+2\underset{e\in E(\linkgadget^{\pl}(B,n))}{\max} |J(e)|\right)
        \leq n^2\,\left(1+ 2 N |D| nB \right)
        \leq 4(|D|+1)N^{2.5} B
        \leq K_R,
    \end{align*}
    so the requirement $K \geq n^2\,(1+2\max_e |J(e)|)$ of~\autoref{lm:planarize} holds. Thus, $\linkgadget^{\pl}(B,n)$ faithfully simulates $\linkgadget(B,n)$.
    
    Finally, in order invoke the induction hypothesis for $\Phi(B,D_L,m-1)$ and $\Phi(B+n,D_R,m-1)$, we need to show that $|D_L|\leq 1$ ($|D_R|=1$ from~\eqref{eq:setDR}).
    \begin{align*}
        |D_L|&=\left|\frac{N}{2K_L}(D-20N^2B)\right| &\text{(from~\eqref{eq:setDL})}\\
        &\leq \left|\frac{ND}{800N^{m+4}B^2}\right|+\left|\frac{20N^3B}{800N^{m+4}B^2}\right| &\text{(from~\eqref{eq:setKL})}\\
        &\leq \frac{1}{800}+\frac{1}{40} \ll 1. &\text{(since $|D|\leq 1$)}
    \end{align*}
    Thus, we can work under the assumption that $\Phi(B,D_L,m-1)$ and $\Phi(B+n,D_R,m-1)$ hold. We may view $G$ as (see~\autoref{fig:glgmgr}) $$ G=\GL\circ\GM\circ\linkgadget^{\pl}\circ\GR, $$ where $\circ$ represents concatenation of graphs. To show that $\Phi(B,D,m)$ holds, we will first show through calculations that properties (i), (ii) hold. To verify that properties (iii), (iv), (v), (vi) hold, we will exhibit $Bn^m$ paths in $G$. For $0 \leq j \leq n^m-1$, write $j=nd +r$ with $0\leq d \leq n^{m-1}-1$ and $0\leq r \leq n-1$; then for $0\leq b \leq B-1$, let
    \begin{equation}\label{eq:Pbj}
    P_{bj} = P_{bd}^L \circ (P_{bd}^L)^{\rev} \circ \linkedge(b,b+r+1) \circ P_{b+r+1,d}^R,
    \end{equation}
    where $\linkedge(b,b+r+1)$ is the unique shortest path (the straight line) in $\linkgadget^{\pl}$ connecting vertex $b$ in the last layer of $\GM$ to vertex $b+r+1$ in the first layer of $\GR$. We will show that in the interval $\interval(j,m)$, the path $P_{bj}$ as defined by~\eqref{eq:Pbj} is the shortest path from the input vertex $b$ in $G$. We are now set to show that properties (i) through (vi) hold for $\Phi$.
    
    \paragraph{Property (i):} Note that the planarization of the linking gadget $\linkgadget^{\pl}(B,n)$ adds at most $(B+n)^4$ new vertices. This means that the number of vertices in the planarized version of $G$ is at most
    \[
        \underset{\GL,\GM}{\underbrace{2(3^m-1)(B+(m-1)n)^4}} +
        \underset{\linkgadget^{\pl}(B,n)}{\underbrace{(B+n)^4}} +
        \underset{\GR}{\underbrace{(3^m-1)(B+n+(m-1)n)^4}} \leq (3^{m+1}-1)(B+mn)^4.
    \]
    Thus, property (i) holds. We now verify property (ii).
    
    \paragraph{Property (ii):}
    
    Using the induction hypothesis, we know that each edge $e$ in the graph $G(B',D',m-1)$ has the form $a_e+b_e\lambda$, where $a_e=a_{1e}+a_{2e}D$, $b_e=b_{1e}+b_{2e}D$. Also,
    \begin{align*}
        \max_e&\{|\num(a_{1e})|,|\num(a_{2e})|,|\num(b_{1e})|,|\num(b_{2e})|\}\leq (400\,NB')^{5(m-1)^2},\\ \max_e&\{|\den(a_{1e})|,|\den(a_{2e})|,|\den(b_{1e})|,|\den(b_{2e})|\}\leq 2^{m-1}n^2,
    \end{align*}
    where $e$ ranges over all the edges of $G(B',D',m-1)$, $\num$ stands for numerator, and $\den$ stands for denominator. Each edge of $G$ comes from either $\GL$, $\GM$, $\linkgadget^{\pl}(B,n)$ or $\GR$.
    
    First we consider edges coming from $\GL$ (we do not consider $\GM$ separately it has the same edge weights as $\GL$). Let $e$ be an edge of $G$ coming from $\GL$. Using the induction hypothesis,
    \begin{align*}
        a_e^L&=a_{1e}^L+a_{2e}^L\left(\frac{N}{2K_L}\left(D-\frac{K_R}{N}\right)\right). &\text{(before scaling)}\\
        b_e^L&=b_{1e}^L+b_{2e}^L\left(\frac{N}{2K_L}\left(D-\frac{K_R}{N}\right)\right). &\text{(before scaling)}
    \end{align*}
    However, once $\GL$ becomes a part of $G$, $a_e^L$ is scaled by $K_L$ and $b_e^L$ is scaled by $K_L/N$.
    \begin{align*}
        a_e&=K_La_{1e}^L+\frac{a_{2e}^LN}{2}\left(D-\frac{K_R}{N}\right) &\text{(after scaling)}\\
        &=400\,N^{m+4}B^2a_{1e}^L+\frac{a_{2e}^LN}{2}\left(D-\frac{20\,N^3B}{N}\right) &\text{(using~\eqref{eq:setKL})}\\
        &=\underset{a_{1e}}{\underbrace{\left(400\,N^{m+4}B^2a_{1e}^L-10\,a_{2e}^LN^3B\right)}}+\underset{a_{2e}}{\underbrace{\left(\frac{a_{2e}^LN}{2}\right)}}D.\\
        b_e&=\frac{K_Lb_{1e}^L}{N}+\frac{b_{2e}^L}{2}\left(D-\frac{K_R}{N}\right) &\text{(after scaling)}\\
        &=\frac{400\,N^{m+4}B^2b_{1e}^L}{N}+\frac{b_{2e}^L}{2}\left(D-\frac{20\,N^3B}{N}\right) &\text{(using~\eqref{eq:setKL})}\\
        &=\underset{b_{1e}}{\underbrace{\left(400\,N^{m+3}B^2b_{1e}^L-10\,b_{2e}^LN^2B\right)}}+\underset{b_{2e}}{\underbrace{\left(\frac{b_{2e}^L}{2}\right)}}D.
    \end{align*}
    Thus, $a_e$ and $b_e$ have the required form. If the denominators of $a_{1e}^L,a_{2e}^L,b_{1e}^L$ and $b_{2e}^L$ have absolute value at most $2^{m-1}n^2$, then the denominators of $a_{1e},a_{2e},b_{1e}$ and $b_{2e}$ have absolute value at most $2^mn^2$. Now we need to check for the numerators.
    \begin{align*}
        |\num(a_{1e})|&=|\num(400\,N^{m+4}B^2a_{1e}^L-10\,a_{2e}^LN^3B)|\\
        &\leq \left|400\,N^{m+4}B^2(400NB)^{5(m-1)^2}\right|+\left|10\,(400NB)^{5(m-1)^2}N^3B\right|\\
        &\leq \left|(200NB)^{10m-5}(400NB)^{5(m-1)^2}\right|+\left|(200NB)^{10m-5}(400NB)^{5(m-1)^2}\right|\\
        &\leq (400NB)^{5m^2}.
    \end{align*}
    We skip the proof for $a_{2e},b_{1e}$ and $b_{2e}$. Now we consider edges coming from $\linkgadget^{\pl}(B,n)$. Let $(b,b+r)\in E(\linkgadget(B,n))$.
    \begin{align*}
        w_{b,b+r} &= NDrb + \frac{K_R}{N}\left(\left(\frac{r(r+1)}{2}\right)N - r\lambda\right) &\text{(using~\eqref{eq:linkedges})}\\
        &= NDrb + \frac{r(r+1)K_R}{2} - \frac{rK_R}{N}\lambda\\
        &= \underset{a_e}{\underbrace{(NDrb + 10\,r(r+1)N^3B)}} + \underset{b_e}{\underbrace{(-20\,rN^2B)}}\lambda.
    \end{align*}
    Note that all these coefficients are integers. However, these are the edge weights from the \emph{non-planar} linking gadget. Once we planarize it, the edges in the \emph{planar} linking gadget can have denominators at most $n^2$ (see~\autoref{defacto}). As for the numerator,
    \begin{align*}
        |NDrb + 10\,r(r+1)N^3B|&\leq NDnB+10N^4B \qquad\quad(\text{since }0\leq b\leq B-1, 0\leq r\leq n, N=n^2)\\
        &\leq (400NB)^{5m^2}.
    \end{align*}
    Finally we consider edges coming from $\GR$. Let $e$ be an edge of $G$ coming from $\GR$. Using the induction hypothesis, we have the following.
    \begin{align*}
        a_e^R&=a_{1e}^R+a_{2e}^RD_R. &\text{(before scaling)}\\
        b_e^R&=b_{1e}^R+b_{2e}^RD_R. &\text{(before scaling)}\\
        a_e&=K_Ra_{1e}^R+K_Ra_{2e}^R &\text{(after scaling)}\\
        &=\underset{a_{1e}}{\underbrace{20\,N^3Ba_{1e}^R + 20\,N^3Ba_{2e}^R}}. &\text{(using~\eqref{eq:setKR})}\\
        b_e&=\frac{K_R}{N}b_{1e}^R+\frac{K_R}{N}b_{2e}^R &\text{(after scaling)}\\
        &=\underset{b_{1e}}{\underbrace{20\,N^2Ba_{1e}^R + 20\,N^2Ba_{2e}^R}}. &\text{(using~\eqref{eq:setKR})}
    \end{align*}
    Thus, $a_e=a_{1e}+0\cdot D$ and $b_e=b_{1e}+0\cdot D$ have the required form. Also, $a_{1e}$ and $b_{1e}$ are integers.
    \begin{align*}
        |\num(a_{1e})|&=|20\,N^3Ba_{1e}^R + 20\,N^3Ba_{2e}^R|\\
        &\leq |20\,N^3B(400NB)^{5(m-1)^2}|+|20\,N^3B(400NB)^{5(m-1)^2}|\\
        &\leq (400NB)^{5m^2}.
    \end{align*}
    We skip the proof for $b_{1e}$. This finishes the verification of property (ii).

    \paragraph{Properties (v), (vi):} Given our definition of $P_{bj}$~\eqref{eq:Pbj}, properties (v) and (vi) are straightforward to verify. We now verify property (iv).
    
    \paragraph{Property (iv):} In the subsequent calculations, we use the following notation. Paths of $G$ are composed of paths coming from $\GL$, $\GM$ and $\GR$; we use $\Cost_L$, $\Cost_M$ and $\Cost_R$ to denote the costs of those subpaths in their constituent graphs. For example, $\Cost_L(P_{bd}^L)(\lambda)$ denotes the cost of the path $P_{bd}^L$ as a function of $\lambda$ in the graph $\GL$. When $\GL$ is used as a component in $G$, this cost is scaled by a factor of $K_L$ and stretched by a factor of $N$. Thus, the cost of the $P_{bj}$ in the graph $G$ is given by
    \begin{align*}
        \Cost(P_{bj})(\lambda) &= K_L\Cost_L(P_{bd}^L)\left(\frac{\lambda}{N}\right) + K_L\Cost_M((P_{bd}^L)^{\rev})\left(\frac{\lambda}{N}\right) + w_{b,b+r+1} + K_R\Cost_R(P_{b+r+1,d}^R)\left(\frac{\lambda}{N}\right)\\
        &= 2K_L\Cost_L(P_{bd}^L)\left(\frac{\lambda}{N}\right) + w_{b,b+r+1} + 
        K_R\Cost_R(P_{b+r+1,d}^R)\left(\frac{\lambda}{N}\right)\\
        &=2K_L\left[\Cost_L(P_{0d}^L)\left(\frac{\lambda}{N}\right) + bD_L\alpha(d,m-1)\right]\\
        &\quad + ND(r+1)b +\frac{K_R}{N}\left[\frac{(r+1)(r+2)}{2} N - (r+1)\lambda \right] \\
        &\quad + K_R\left[\Cost_R(P_{0d}^R)\left(\frac{\lambda}{N}\right) + (b+r+1)D_R \alpha(d,m-1) \right].
    \end{align*}
    Substitute $b=0$ to get 
    \begin{align*}
        \Cost(P_{0j})(\lambda)=&\ 2K_L\Cost_L(P_{0d}^L)\left(\frac{\lambda}{N}\right)+\frac{K_R}{N}\left[\frac{(r+1)(r+2)}{2} N - (r+1)\lambda \right]\\
        &+ K_R\Cost_R(P_{0d}^R)\left(\frac{\lambda}{N}\right) + K_R(r+1)D_R \alpha(d,m-1).
    \end{align*}
    With this expression for $\Cost(P_{0j})(\lambda)$, we obtain
    \begin{align*}
        \Cost(P_{bj})(\lambda)&= \Cost(P_{0j})(\lambda) + b\left[2K_L D_L \alpha(d,m-1) + K_RD_R\alpha(d,m-1) + ND(r+1)\right]\\
        &= \Cost(P_{0j})(\lambda) + b\left[2K_L\ \frac{N}{2K_L}\left(D-\frac{K_R}{N}\right) \alpha(d,m-1) + K_R \alpha(d,m-1) + ND(r+1)\right]\\
        &= \Cost(P_{0j})(\lambda) + b\left[ND\alpha(d,m-1) - K_R \alpha(d,m-1) + K_R \alpha(d,m-1) + ND(r+1)\right]\\
        &= \Cost(P_{0j})(\lambda) + bD\left[N\alpha(d,m-1) + N(r+1)\right]\\
        &= \Cost(P_{0j})(\lambda) + bD \alpha(j,m).
    \end{align*}
    Thus, property (iv) also holds. All that remains is to verify property (iii).
    
    \paragraph{Property (iii):} To verify property (iii), we need to check that $P_{bj}$ as defined above is indeed the shortest path from input vertex $b$ to the last layer when $\lambda \in\interval(j,m)$, and any deviation from it attracts significant additional cost. We do this through two claims (\autoref{cl:LM} and \autoref{cl:LR}).
    
    In \autoref{cl:LM}, we track paths from an input vertex as they travel through $\GL$ and $\GM$. In \autoref{cl:LR}, we analyze how such paths continue through $\linkgadget^{\pl}$ and $\GR$. Fix a $j$ ($0\leq j \leq n^m-1$, say $j=nd + r$, for $0\leq d \leq n^{m-1}-1$ and $0\leq r\leq n-1$) and a $\lambda\in\interval(j,m)$. Note that since $\lambda \in \interval(j,m)$, we have $\lambda/N \in \interval(d,m-1)=[\alpha(d,m-1)+1, \alpha(d,m-1)+N-1]$.
    \begin{Claim} \label{cl:LM}
    Let $Q$ be a path from the input vertex $b$ to the last layer of $\GL\circ\GM$ $($note that $P^L_{bd} \circ (P^L_{bd})^{\rev}$ is one such path$)$. If $Q \neq P^L_{bd} \circ (P^L_{bd})^{\rev}$, then
    \[
        \Cost(Q)(\lambda) - \Cost(P^L_{bd} \circ (P^L_{bd})^{\rev})(\lambda) \geq K_L/2.
    \]
    \end{Claim}
    
    \begin{proof}[Proof of claim] We omit the argument $\lambda$ in this discussion. Let $Q = Q^L \circ Q^M$, where $Q^L$ is the subpath of $Q$ in $\GL$ and $Q^M$ is the subpath of $Q$ in $\GM$. Suppose $Q^M$ terminates at vertex $c$ in the last layer of $\GM$. Then,
    \begin{align*}
    \Cost(Q) - \Cost(P^L_{bd} \circ (P^L_{bd})^{\rev})
    &= \left(\Cost(Q^L) + \Cost(Q^M)\right) - \left(\Cost(P^L_{bd}) + \Cost ((P^L_{bd})^{\rev})\right)\\
    &= \left(\Cost(Q^L) - \Cost(P^L_{bd})\right) +
    \left(\Cost(Q^M) - \Cost(P^L_{cd})\right) +
    \left(\Cost((P^L_{cd})^{\rev}) - \Cost ((P^L_{bd})^{\rev})\right)\\
    &\geq  \overset{\mathrm{Term\ I}}{\overbrace{\Cost(Q^L) - \Cost(P^L_{bd})}} + \overset{\mathrm{Term\ II}}{\overbrace{\Cost(Q^M) - \Cost(P^L_{cd})}}
    - \overset{\mathrm{Term\ III}}{\overbrace{K_LBD_L\alpha(d,m-1)}}.
    \end{align*}
    To obtain $\mathrm{Term\ III}$, we use part (ii) of the induction hypothesis for $\GL$, whose edge costs we evaluated at $\lambda/N$ and scaled by $K_L$; recall that $\lambda/N \in \interval(d,m-1)$.
    If $Q \neq P^L_{bd} \circ (P^L_{bd})^{\rev}$, then one of the following is true.
    \begin{enumerate}
        \item[$($a$)$] $Q^L \neq P^L_{bd}$;
        \item[$($b$)$] $c=b$ and $Q^M \neq (P^L_{bd})^{\rev}$;
        \item[$($c$)$] $c\neq b$ and $Q^M \neq (P^L_{cd})^{\rev}$ (here we use the fact that the paths $P^L_{bd}$ and $P^L_{cd}$ are vertex-disjoint if $c\neq b$).
    \end{enumerate}
    From property (iii) of the induction hypothesis, the costs of a shortest and a non-shortest path from the same input vertex differ by at least one in $\GL$ and $\GM$; after scaling all the edge weights of $\GL$ and $\GM$ by a factor of $K_L$, this difference becomes at least $K_L$. Also note that both $\mathrm{Term\ I}$ and $\mathrm{Term\ II}$ are non-negative. Thus we can conclude the following.
    
    If (a) is true, then $\mathrm{Term\ I} \geq K_L$. If (b) or (c) is true, then $\mathrm{Term\ II} \geq K_L$. Recall that $\alpha(d,m-1)\leq N^m$. For the setting of $K_L$ according to~\eqref{eq:setKL}, we have $\left|\mathrm{Term\ III}\right|=\left|K_LBD_L\alpha(d,m-1)\right|\ll K_L/10$. This completes the proof of~\autoref{cl:LM}.
    \end{proof}
    Since $K_L$ is positive, \autoref{cl:LM} implies that $P^L_{bd} \circ (P^L_{bd})^{\rev}$ is the shortest path from the input vertex $b$ to the last layer of $\GL \circ \GM$. Now, we need to argue that the overall shortest path must be an extension of this. The next claim shows that the shortest path from an input vertex $b$ of $\linkgadget(B,n)$ (note that is the non-planar version of the linking gadget) in the graph $\linkgadget(B,n)\circ\GR$ must follow the route prescribed by~\eqref{eq:Pbj}.
    \begin{Claim} \label{cl:LR}
    Let $\lambda \in \interval(j,m)$, where $j=nd+r$ $(0\leq d \leq n^{m-1}-1$ and $0\leq r \leq n-1)$. Let $P$ be a path from the input vertex $b$ of $\linkgadget(B,n)$ to the last layer of $\GR$ $($note that $\linkedge(b,b+r+1) \circ P^R_{b+r+1,d}$ is one such path$)$. If $P \neq \linkedge(b,b+r+1) \circ P^R_{b+r+1,d}$, then 
    $$ \Cost(P)(\lambda) - \Cost(\linkedge(b,b+r+1) \circ P^R_{b+r+1,d})(\lambda) \geq 1.$$
    \end{Claim}
    \begin{proof}[Proof of claim] Fix the input vertex $b$. The induction hypothesis guarantees that $P^R_{xd}$ is the unique shortest path from the input vertex $x$ of $\GR$ to the last layer of $\GR$. We may assume that $P$ travels travels along the shortest path in $\GR$, that is, it has the form
    \[ P_k = \linkedge(b,b+k) \circ P^R_{b+k,d},\]
    for some $k \in \{0,1,\ldots,n\}$. Let $Z_k = \Cost(P_k)(\lambda)$. We will show that for $\lambda \in \interval(j,m)$, we have
    \begin{equation}
        Z_0 \gg Z_1 \gg \cdots \gg Z_r \gg Z_{r+1} \ll Z_{r+2} \ll \cdots \ll Z_{n}, \label{eq:downandup}
    \end{equation}
    where we use $\gg$ and $\ll$ to suggest that there is a large gaps between the quantities. Our proof strategy is to compare successive values of $Z_k$. We will show that $Z_k - Z_{k-1}$ is negative whenever $k\leq r+1$ and positive otherwise. Indeed, for $k=1,2,\ldots,n$, we have
    \begin{align*}
    Z_k - Z_{k-1} &= w_{b,b+k} - w_{b,b+k-1} + \Cost(P^R_{b+k,d})(\lambda) - \Cost(P^R_{b+k-1,d})(\lambda),\\
    \intertext{where}
    w_{b,b+k} &= ND kb  + \frac{K_R}{N}\left(\left(\frac{k(k+1)}{2}\right) N - k \lambda\right);\\
    w_{b,b+k-1} &= ND (k-1)b + \frac{K_R}{N}\left(\left(\frac{(k-1)k}{2}\right)N - (k-1)\lambda\right);\\
    \Cost(P^R_{b+k,d})(\lambda) &= K_R\left[\Cost_R(P^R_{0,d})(\lambda/N) + (b+k)D_R\alpha(d,m-1)\right];\\
    \Cost(P^R_{b+k-1,d})(\lambda)&= K_R\left[\Cost_R(P^R_{0,d})(\lambda/N) + (b+k-1)D_R\alpha(d,m-1)\right].
    \end{align*}
    Thus,
    \begin{align}
    Z_k-Z_{k-1}&= NDb + \frac{K_R}{N}(kN - \lambda) + K_R D_R \alpha(d,m-1)\\
    &= NDb +\frac{K_R}{N}(kN + N\alpha(d,m-1) -\lambda) & \text{(recall $D_R=1$)}\\
    &= NDb + \frac{K_R}{N}(\alpha(k-1,m) -\lambda). \label{eq:dominate}
    \end{align}
    Since $\lambda \in \interval(r,m) = [\alpha(r,m)+1, \alpha(r,m)+N-1]$, we have
    \begin{align*}
        \alpha(k-1,m) -\lambda
        & \in [\alpha(k-1,m) - \alpha(r,m) - N +1, \alpha(k-1,m) - \alpha(r,m) -1] \\
        & = [(k-(r+1))N - N + 1, (k-(r+1))N -1]\,.
    \end{align*}
    Thus, for $k=1,2,\ldots, r+1$, we have $\alpha(k-1,m) -\lambda \leq -1$ and
    for $k=r+2, \ldots, n$, we have $\alpha(k-1,m) -\lambda \geq +1$. Returning to~\eqref{eq:dominate} with this, we obtain
    \begin{align}
        Z_k-Z_{k-1} & \leq NDb - \frac{K_R}{N} &&\text{for $k=1,\ldots,r+1$, and} \label{eq:smallk}\\
        Z_k-Z_{k-1} & \geq NDb +\frac{K_R}{N} &&\text{for $k=r+2,\ldots,n$}. \label{eq:bigk}
    \end{align}
    Since $K_R \gg N^2b$ and $-1\leq D\leq +1$, the RHS of~\eqref{eq:smallk} is negative and the RHS of~\eqref{eq:bigk} is positive. This confirms~\eqref{eq:downandup} and establishes~\autoref{cl:LR}.
    \end{proof}
    We are now in a position to establish property (iii) and complete the induction. By~\autoref{cl:LM}, if the shortest path from $b$ deviates from $P^L_{bd} \circ (P^R_{bd})^{\rev}$ in $\GL\circ\GM$, then the increase in cost is at least $K_L/2$. We will show that the difference in cost between every two paths originating at an input vertex of $\linkgadget(B,n)$ and terminating in the last layer of $\GR$ is much smaller than $K_L/2$; this forces the shortest path from the input vertex $b$ in $G$ 
    when restricted to $G_L \circ G_M$ to be precisely $P^L_{bd} \circ (P^R_{bd})^{\rev}$. Let $P_1$ and $P_2$ be paths  originating at an input vertex of $\linkgadget(B,n)$ and terminating in the last layer of $\GR$. Then,
    \begin{align*} \label{eq:lastequations}
    |\Cost(P_1)(\lambda) - \Cost(P_2)(\lambda)|
    &\leq \left|NDnB\right|+\left|\frac{K_R}{N}\left(n^2N+n\lambda\right)\right|+\left|K_RDB\alpha(d,m-1)\right|\\
    &\leq \left|N^2DB\right|+\left|\frac{K_R}{N}\left(N^2+n(\alpha(j,m)+N)\right)\right|+\left|K_RDB\alpha(d,m-1)\right|\\
    &\leq \left|N^2DB\right|+\left|\frac{K_R}{N}\left(N^2+n(N^{m+1}+N)\right)\right|+\left|K_RDBN^m\right|\\
    &\leq \left|N^2DB\right|+\left|K_RN^{m+1}\right|+\left|K_RDBN^m\right|\\
    &\ll K_L/10.
    \end{align*}
    Thus, the shortest path in $G$ must follow the path $P^L_{bd} \circ (P^R_{bd})^{\rev}$ until it arrives at the first layer of $\linkgadget(B,n)$; for if it does not, then its cost is at least $K_L/2 - K_L/10 \gg 1$ more than the cost of $P_{bd}$. \autoref{cl:LR} now confirms that it must continue by taking the edge $\linkedge(b,b+r+1)$ and $P^R_{b+r+1,d}$; any deviation from this path will incur an increase in cost of at least $1$. Therefore, the shortest path in $G$ in the interval $\interval(j,m)$ is $P_{bj}$ as promised~\eqref{eq:Pbj}. This completes the proof of property (iii), hence completing the proof of~\autoref{lm:mainlemma}.
    \end{proof}
    
    \section{Upper bound for polynomially varying edge weights} \label{sec:complpoly}
    
    In this section, we consider graphs with edge weights of the form $$w_e(\lambda)=a_d\lambda^d+a_{d-1}\lambda^{d-1}+\cdots+a_1\lambda+a_0.$$ The main result of this section is the following. 
    
    \thmcomplpoly*
    
    
    \begin{Definition} \label{def:ackermann}
    The Ackermann function, first used by Wilhelm Ackermann~\cite{wackermann}, is defined as follows.
    \begin{align*}
        A(m, n) =
        \begin{cases}
            n+1 & \mbox{if } m = 0 \\
            A(m-1, 1) & \mbox{if } m > 0 \mbox{ and } n = 0 \\
            A(m-1, A(m, n-1)) & \mbox{if } m > 0 \mbox{ and } n > 0.
        \end{cases}
    \end{align*}
    The inverse Ackermann function, denoted by $\alpha(n)$, is the minimum $r$ such that $A(r,r)\geq n$.
    \end{Definition}
    
    For our proof, we require the definition of Davenport-Schinzel sequences, which will help us bound the number of alternations between pairs of paths.
    \begin{Definition}
        Given a finite set of symbols $\cX$, a sequence $U = (u_1,u_2,u_3,\ldots,u_n)$ is a Davenport-Schinzel sequence of order $s$ if it satisfies the following properties:
        \begin{itemize}
            \item Each $u_i$ $($for $i\in[n])$ is a symbol coming from $\cX$.
            \item No two consecutive symbols in the sequence $U$ are the same.
            \item If $x_1\in\cX$ and $x_2\in\cX$ are two distinct symbols, then $U$ does not contain a subsequence $(\ldots,x_1,\ldots,x_2,\ldots,x_1,\ldots,x_2,\ldots)$ consisting of $s+2$ alternations between $x_1$ and $x_2$.
        \end{itemize}
        Then we say that $U$ is a $DS(n,s)$-sequence.
    \end{Definition}
    
    \begin{proof} [Proof of~\autoref{thm:complpoly}]
    We adapt to our setting an argument due to Gusfield~\cite[Page 100]{carstensenthesis}. Fix an integer $d>1$ and consider only those graphs whose edge weights are polynomials of degree at most $d$. Let $f(n,\ell)$ be the maximum length of a sequence of shortest paths\footnote{Note that this sequence might have alternations. That is, a shortest path can occur more than once in this sequence (however, two consecutive shortest paths must be distinct).}, when the paths are restricted to have at most $\ell$ edges. Fix a sequence $\sigma$ of paths. Let $p$ be a path in $\sigma$. We may fix a vertex $v$ in $p$ such that $v$ is the \emph{middle} vertex of the path $p$. That is, $p$ has at most $\ceil{\ell/2}$ edges from $s$ to $v$, and at most $\ceil{\ell/2}$ edges from $v$ to $t$. Then, the number of such paths in $\sigma$ that pass through $v$ is at most $2f(n,\ceil{\ell/2})$. Accounting for all $v$, we obtain that there are at most $2nf(n,\ceil{\ell/2})$ paths in the sequence $\sigma$. Now, $\sigma$ might have alternations. Thus, if $N$ is the number of \emph{distinct} paths in $\sigma$, then $N \leq 2n f(n,\ceil{\ell/2})$. Since the costs of these paths are polynomials of degree at most $d$ in $\lambda$, two paths can alternate at most $d+1$ times (the curves of two distinct degree $d$ polynomials cannot intersect each other in more than $d$ points). That is, $\sigma$ is a Davenport-Schinzel sequence of order $d$ with an alphabet of size $N$. Bounds known for Davenport-Schinzel sequences (see Matou\v{s}ek~\cite[Page 173]{matousek}) imply that the maximum length of such a sequence of shortest paths is at most $N 2^{\alpha(N)^{d}}$ (for all large $N$).
    
    Since $N \leq n^n$ (a coarse upper bound on the total number of paths in an $n$-vertex graph), we have $\alpha(N)\ll\alpha(n)+2$. Thus,
    \begin{align*}
    f(n,\ell)&\leq N \cdot 2^{\alpha(N)^{d}}\\
    &\leq 2n f(n,\ceil{\ell/2}) \cdot 2^{(\alpha(n)+2)^d}\\
    &\leq (2n)^{\log \ell} \cdot f(n,1) \cdot \left(2^{(\alpha(n)+2)^d}\right)^{\log \ell}.
    \end{align*}
    Our theorem follows from this after substituting $\ell=n$ and $f(n,1)\leq 1$.
    \end{proof}
    
    \paragraph{Remarks:} (i) Note that even though the edge weights are polynomials of degree $d$ in the parameter $\lambda$, the upper bound is not significantly higher than that for $d=1$. (ii) Since our proof counts each time a path reappears as separate occurrence, the same upper bound holds even if we count the number of different paths \emph{with multiplicity}. (iii) Our proof works for any family of functions $\mathcal{F}$ which satisfy the following conditions (the family of polynomials clearly satisfies these conditions).
    \begin{enumerate}
        \item $\mathcal{F}$ is closed under addition.
        \item For all $f_1,f_2\in\mathcal{F}$, the sign of $f_1-f_2$ can change only a small number of times.
    \end{enumerate}
    
    \section{Parametric shortest paths with three parameters} \label{sec:complmult}
    
    In this section, we consider graphs with edge weights of the form $$w_e(\lambda_1,\lambda_2,\lambda_3)=a_e\lambda_1 + b_e \lambda_2 + c_e\lambda_3.$$  The main result of this section is the following.
    \thmcomplmult*
    Before we jump into our proof, we note that if the edge weights have the form $a_e\lambda_1+b_e\lambda_2$, then it reduces to the single parameter case. Consider the paths restricted to $\lambda_2=-1$ and $\lambda_2=1$. It is easy to see that a path that forms a lower envelope without this restriction continues to be a part of the lower envelope for either $\lambda_2=-1$ or for $\lambda_2=1$ (or for both). Therefore, we can match Gusfield's bound (up to a factor $2$), which is for edge weights of the form $a_e\lambda+b_e$. Thus, $\complmult{2}(n)=n^{\log n+O(1)}$. Let us now commence our proof for the three parameter case.
    
    
    Note that the cost of an $s$-$t$ path $P$ has the form $\Cost(P)(\lambda_1,\lambda_2,\lambda_3)=a_P\lambda_1+b_P\lambda_2+c_P\lambda_3$. We may view these paths as vectors in $\mathbb{R}^3$. That is, $\Cost(P)(\vlambda)=\va_P \cdot \vlambda$, where $\vlambda=(\lambda_1,\lambda_2,\lambda_3)\in\mathbb{R}^3$ and $\va_P=(a_P,b_P,c_P)\in\mathbb{R}^3$. We focus on the coefficient vectors $\va_P$.
    Let $\conv(G)\subseteq \mathbb{R}^3$ be the convex hull of the 
    set $\{\va_P: \text{$P$ is an $s$-$t$ path in $G$}\}$. The crucial observation is that for each $\vlambda \in \mathbb{R}^3$, at least one of the vertices (vertex here means $0$-dimensional face) of $\conv(G)$ corresponds to a minimum cost path. Thus, the number of vertices of $\conv(G)$ is an upper bound on $\complmult{3}(G)$.
    
    To upper bound the number of vertices of $\conv(G)$, we proceed by induction roughly as in earlier proofs, but working in the language of polytopes. As before, fix $\ell$ and let $\conv(G,\ell)$ be the convex hull of the set $\{\va_P: \text{$P$ is an $s$-$t$ path in $G$ with $\ell$ edges}\}$. We view every $s$-$t$ path $P$ with $\ell$ edges as consisting of an $s$-$v$ path $P_{sv}$ and a $v$-$t$ path $P_{vt}$, each with about $\ell/2$ edges. Then the coefficient vector of the path $P$ is the sum of the corresponding vectors of $P_{sv}$ and $P_{vt}$, that is,
    $\va_{P}=\va_{P_{sv}} + \va_{P_{vt}}$. This naturally leads us to consider the two polytopes $\conv(G_{sv},\ell/2)$ and $\conv(G_{vt},\ell/2)$ generated by coefficients of the form $\va_{P_{sv}}$ and $\va_{P_{vt}}$, respectively. Thus, the convex hull of the coefficient vectors of paths bisected at $v$ is the Minkowski sum $\conv(G_{sv},\ell/2) + \conv(G_{vt},\ell/2)$. 
    
    We show below how the number of vertices in such a Minkowski sum can be bounded. The final bound is obtained by running over all choices of $v$. Before presenting the formal proof, we show some facts about convex polytopes that we will need for our proof (we follow notation used by Michel Goemans in his lecture notes~\cite{goemans}).
    
    
    \subsection{Convex polytopes and Minkowski sums}
    
    Let $\cA \subseteq \mathbb{R}^k$ be a non-empty convex polytope. The dimension of $\cA$, denoted by $\dim(\cA)$, is the maximum number of affinely independent points in $\cA$ minus $1$. A face of $\cA$ is a set of the form $\{\vx \in \cA: \vlambda \cdot \vx = \beta\}$, where $\vlambda \cdot \vx \geq \beta$ for all $\vx \in \cA$. Let $\cF({\cA})$ be the set of faces of $\cA$, and $\cF^{\,r}(\cA)$ be the set of faces of $\cA$ of dimension $r$, where $0\leq r\leq k$ (note that a face is also a polytope). Let $\compl_r(\cA) = |\cF_{\cA}^{\,r}|$ and $\compl(\cA)=|\cF_{\cA}|$.
    
    The affine space spanned by the face $F$ has the form $x_0 + \tS_F$, where $x_0 \in F$ and $\tS_F$ is a subspace of dimension $\dim(F)$. We use $\Pi_F$ to denote the orthogonal projection onto the space $\tS_F^{\perp}$, which is the subspace of $\mathbb{R}^d$ orthogonal to the subspace $\tS_F$.
    
    We say that a face $F$ is non-trivial if $0\neq\dim(F)\neq\dim(\cA)$. Then, every non-trivial face of $\cA$ is the set of solutions for a linear program of the form $\langle\,\min \vlambda \cdot \vx$; subject to $\vx \in \cA\,\rangle$ with $\vlambda \neq 0$; we use $\sols(\cA;\vlambda)$ to refer to this set of solutions.
    
    For polytopes $\cA$ and $\cB$, let $\cA+\cB$ denote their Minkowski sum, that is, $\cA+\cB= \{\vx+\vy: \vx\in \cA \text{ and } \vy \in \cB\}$.
    \begin{Proposition}[{see~\cite[Proposition 2.1]{fukuda}}]\label{prop:facedecomposition}
        Every face $F$ of $\cA+\cB$ can be written uniquely as $F=F_{\cA}+F_{\cB}$, where $F_{\cA}  \in \cF(\cA)$ and $F_{\cB} \in \cF({\cB})$. In particular, if 
        \begin{equation} \label{eq:facelp}
            F=\sols(\cA+\cB;\vlambda),
        \end{equation} 
        then the unique decomposition is $F=\sols(\cA;\vlambda)+\sols(\cB;\vlambda)$; in particular, $\sols(\cA;\vlambda)$ and $\sols(\cB;\vlambda)$ do not depend on the choice of $\vlambda$ in~\eqref{eq:facelp}.
    \end{Proposition}
    
    For a face $F_{\cA}$ of $\cA$, let $\cF({\cA,\cB|F_\cA})= \{F \in \cF({\cA+\cB}): F = F_{\cA} + F_{\cB} \text{ for some }F_{\cB} \in \cF({\cB})\}$, and let $\compl(\cA,\cB|F_{\cA}) = |\cF({\cA,\cB|F_\cA})|$; that is, $\compl(\cA,\cB|F_{\cA})$ is the number of faces of $\cA + \cB$ whose unique decomposition (as in~\autoref{prop:facedecomposition}) involves $F_{\cA}$.
    
    \begin{Lemma} \label{lm:savingface}
        Suppose $F_{\cA}$ is non-empty. Then,  $\compl(\cA,\cB|F_{\cA}) \leq \compl(\Pi_{F_{\cA}}(\cB))$.
    \end{Lemma}
    \begin{proof}
    We will exhibit a one-to-one (injective) map from $\cF({\cA,\cB|F_\cA})$ to $\cF(\Pi_{F_{\cA}}(\cB))$. To simplify notation for this proof, we will just write $\Pi$ instead of $\Pi_{F_{\cA}}$. Fix a face $F=F_{\cA} + F_{\cB}$ of the polytope $\cA+\cB$. Let $\vlambda$ be such that $F=\sols(\cA+\cB; \vlambda)$, $F_{\cA} = \sols(\cA; \vlambda)$ and $F_{\cB}= \sols(\cB; \vlambda)$; note that $\vlambda \in \tS_{F_{\cA}}^{\perp}$.

    \begin{Claim} $\Pi(F_{\cB}) = \Pi(\sols(\cB;\vlambda)) = \sols(\Pi(\cB);\vlambda)$; in particular, $\Pi(F_{\cA})$ is a face of $\Pi(\cB)$.
    \end{Claim}
    \begin{proof} The first equality follows directly by our choice of $\vlambda$ above. Let $\beta = \vlambda \cdot \vx$ for $\vx \in F_{\cB}$. To justify the second equality, first observe that for all $\vx$, we have
    \begin{equation}
        \vlambda \cdot \vx = \vlambda \cdot \Pi(\vx) \qquad\qquad (\text{since } \vlambda \in \tS_{F_{\cA}}^{\perp}). \label{eq:projection}
    \end{equation}
    Then for all $\vx \in \cB$, we have $\vlambda \cdot \Pi(\vx) = \vlambda \cdot x \geq \beta$. Also, if $\vx \in \sols(\cA;\vlambda)$, then $\vlambda \cdot \Pi(\vx) = \beta$. Hence,  $\Pi(\sols(\cA;\vlambda)) \subseteq \sols(\Pi(\cB);\vlambda)$. Next suppose $\vy \in \sols(\Pi(\cB);\vlambda)$, then $\vy = \Pi(\vz)$ for some $\vz \in \cB$. By~\eqref{eq:projection}, we have $\vlambda \cdot \vz = \vlambda \cdot \vy = \beta$. Thus, $\vz \in \sols(\cB;\vlambda)$ and $\vy \in \Pi(S(\cB;\vlambda))$. Hence, $\sols(\Pi(\cB);\vlambda)\subseteq \Pi(\sols(\cA;\vlambda))$.
    \end{proof}
    
    Consider the map $\eta: \cF(\cA,\cB|F_{\cA}) \rightarrow \cF(\Pi(\cB))$, which maps the face $F = F_{\cA} + F_{\cB}$ to the face $\Pi(F_{\cB})$. By the above claim, this map is well-defined. 
    \begin{Claim}
        $\eta$ is one-to-one.
    \end{Claim}
    \begin{proof} Suppose $\eta(F_{\cA}+F_{\cB})=\eta(F_{\cA}+F_{\cB}')$. Then, by the definition of $\eta$, we have $\Pi(F_{\cB}) = \Pi(F_{\cB}')$. Suppose $F_{\cA} + F_{\cB} = \sols(\cA+\cB; \vlambda)$ and $\vlambda \cdot \vx = \beta$ for all $\vx \in F_{\cB}$. In particular, $\forall\ \vz\in\Pi(F_{\cB})$, we have $\vlambda\cdot\vz=\beta$. Then, for all $\vy \in F_{\cB}'$, we have
    \begin{align*}
        \vlambda \cdot  \vy &= \vlambda \cdot \Pi(\vy) && \text{(by~\eqref{eq:projection})}\\
        &= \beta. && \text{(since
        $\Pi(\vy)\in\Pi(F_{\cB}')=\Pi(F_{\cB})$)}
    \end{align*}
    Thus, $\vy \in \sols(\cB; \vlambda) = F_{\cB}$. Hence, $F_{\cB}' \subseteq F_{\cB}$. Similarly, $F_{\cB} \subseteq F_{\cB}'$.
    \end{proof}
    This completes the proof of~\autoref{lm:savingface}.
    \end{proof}
    
    \subsection{Upper bound on the number of faces}
    
    We now formally carry out the inductive argument outlined earlier. The reason our proof works in three dimensions but not in higher dimensions is as follows. In the three-dimensional world, the number of faces of various dimensions (vertices, edges and faces) in a polytope are within small constant factors of each other, but there is no such bound in higher dimensions.
    
    \begin{proof} [Proof of~\autoref{thm:complmult}]
    Let $\paths$ be a convex polytope in three dimensions. Since each vertex in $\paths$ is connected to at least three other vertices by edges,
    \begin{equation} \label{eq:vert}
        3\compl_0(\paths)\leq 2\compl_1(\paths).
    \end{equation}
    Similarly, since each face in the polytope is adjacent to at least three other faces by edges,
    \begin{equation} \label{eq:face}
        3\compl_2(\paths)\leq 2\compl_1(\paths).
    \end{equation}
    Combining Euler's formula, $\compl_0(\paths) - \compl_1(\paths) + \compl_2(\paths) = 2$, with~\eqref{eq:vert} and~\eqref{eq:face}, we obtain
    \begin{equation} \label{eq:euler}
        2\compl_0(\paths) \leq \compl_1(\paths)+\compl_2(\paths) \leq 5\compl_0(\paths).
    \end{equation}
    
    We now set up some notation in order to state the recurrence that serves as a key element to our proof. Let $r\in\{0,1,2,3\}$ denote the dimension of a face. Let $\complmult{3}_r(n,\ell) = \max_{\,G}\,\compl_r(\conv(G,\ell))$ and $\complmult{3}(n,\ell) = \max_{\,G}\,\compl(\conv(G,\ell))$, where $G$ varies over all graphs on $n$ vertices. Let $\complmult{3}_{1,2}(n,\ell)=\complmult{3}_1(n,\ell)+\complmult{3}_2(n,\ell)$. Similarly, for a convex polytope $\paths$, let $\compl_{1,2}(\paths)=\compl_1(\paths)+\compl_2(\paths)$.
    
    We make the following claim.
    
    \begin{Claim} $$\complmult{3}_{1,2}(n,\ell) \leq 5n\cdot\complmult{2}(n)\cdot\complmult{3}_{1,2}(n,\ell/2).$$
    \end{Claim}

    \begin{proof}
    Fix an $n$ and an $\ell$ such that $1\leq\ell\leq n$. Let $G$ be a graph on $n$ vertices that maximizes $\complmult{3}_{1,2}(n,\ell)$. Let $\cA_{sv}=\conv(G_{sv},\ell/2)$ and $\cB_{vt}=\conv(G_{vt},\ell/2)$. Then, $$\conv(G,\ell) = \conv\left(\bigcup_{\substack{v:\,v \text{ bisects}\\ \text{an $s$-$t$ path}}} (\cA_{sv} + \cB_{vt})\right),$$ where the $+$ stands for the Minkowski sum. By the union bound,
    \begin{align*}
    \compl_0(\conv(G,\ell)) &\leq \sum_{\substack{v:\,v \text{ bisects}\\ \text{an $s$-$t$ path}}}\compl_0(\cA_{sv} + \cB_{vt})\\
    \compl_{1,2}(\conv(G,\ell)) \leq 5\,\compl_0(\conv(G,\ell))&\leq \frac{5}{2}\sum_{\substack{v:\,v \text{ bisects}\\ \text{an $s$-$t$ path}}}\compl_{1,2}(\cA_{sv} + \cB_{vt}). \qquad\qquad \text{(using~\eqref{eq:euler})}
    \end{align*}
    From~\autoref{prop:facedecomposition} and~\autoref{lm:savingface},
    \begin{align*}
        \compl_{1,2}(\conv(G,\ell))&\leq \frac{5}{2} \sum_{\substack{v:\,v \text{ bisects}\\ \text{an $s$-$t$ path}}}\compl_{1,2}(\cA_{sv}+\cB_{vt})\\
        &\leq \frac{5}{2}\sum_{\substack{v:\,v \text{ bisects}\\ \text{an $s$-$t$ path}}}\left(\sum_{\substack{F\,\in\,\cF(\cA_{sv}), \\ \dim(F)\in\{1,2\}}} \compl_{1,2}(\cA_{sv},\cB_{vt}|F) + \sum_{\substack{F\,\in\,\cF(\cB_{vt}), \\ \dim(F)\in\{1,2\}}} \compl_{1,2}(\cB_{vt},\cA_{sv}|F)\right)\\
        &\leq \frac{5}{2}\sum_{\substack{v:\,v \text{ bisects}\\ \text{an $s$-$t$ path}}}\left(\sum_{\substack{F\,\in\,\cF(\cA_{sv}), \\ \dim(F)\in\{1,2\}}} \compl_{1,2}(\Pi_F(\cB_{vt})) + \sum_{\substack{F\,\in\,\cF(\cB_{vt}), \\ \dim(F)\in\{1,2\}}} \compl_{1,2}(\Pi_F(\cA_{sv}))\right).
    \end{align*}
    Recall that $G$ maximizes $\complmult{3}_{1,2}(n,\ell)$. Since both $\Pi_F(\cB_{vt})$ and $\Pi_F(\cA_{sv})$ are projections in one or two dimensions,
    \begin{align*}
    \complmult{3}_{1,2}(n,\ell)&\leq \frac{5}{2}\left(\sum_{\substack{v:\,v \text{ bisects}\\ \text{an $s$-$t$ path}}}\,\overbrace{\complmult{3}_{1,2}(n,\ell/2)}^{\text{Faces from } \cA_{sv}}\cdot\overbrace{\complmult{2}(n)}^{\substack{\text{Faces from}\\ \Pi_F(\cB_{vt})}} + \sum_{\substack{v:\,v \text{ bisects}\\ \text{an $s$-$t$ path}}}\,\overbrace{\complmult{3}_{1,2}(n,\ell/2)}^{\text{Faces from } \cB_{vt}}\cdot\overbrace{\complmult{2}(n)}^{\substack{\text{Faces from}\\ \Pi_F(\cA_{sv})}}\right)\\
    &\leq 5\,\complmult{2}(n)\sum_{\substack{v:\,v \text{ bisects}\\ \text{an $s$-$t$ path}}}\complmult{3}_{1,2}(n,\ell/2).
    \end{align*}
    Since there are at most $n$ choices for the bisecting vertex $v$, summing over all $n$ vertices of the graph establishes our claim.
    \end{proof}
    Using this recurrence, we are now ready to complete our proof. We have $\complmult{2}(n)\leq n^{\log n+O(1)}$ (as in Gusfield's bound). Thus,
    \begin{align*}
        \frac{2}{3}\complmult{3}(n,\ell)\leq\complmult{3}_{1,2}(n,\ell)&\leq \left(5n\cdot n^{\log n+O(1)}\right)\cdot\complmult{3}_{1,2}(n,\ell/2)\\
        &\leq \left(5n\cdot n^{\log n+O(1)}\right)^{\log \ell}\\
        &\leq n^{\log n\log \ell+O(\log \ell)}\\
        \complmult{3}(n)=\complmult{3}(n,n)&\leq n^{(\log n)^2+O(\log n)}.
    \end{align*}
    This completes the proof of~\autoref{thm:complmult}.
    \end{proof}
    
    \section{Alternation-free sequences of paths in graphs}
    
    In this section, we show~\autoref{thm:altfree}.
    
    \thmaltfree*
    
    We construct non-planar and planar graphs with lengthy alternation-free sequences of paths. Our graphs are inspired by earlier examples of non-planar graphs with alternation-free sequences of paths~\cite{carstensenthesis, mulmuleyshah}. For our construction, we introduce a concept of alternation-free sequences of \emph{words} (in the case of planar graphs, these words will be binary strings), where each word corresponds to a path. It so turns out that these words, when arranged in a standard lexicographic order, correspond to an alternation-free sequence of paths. In~\autoref{sec:pathstowords}, we will formally build this connection between paths and words. In~\autoref{sec:words} and~\autoref{sec:binarywords}, we will use this connection to show~\autoref{thm:altfree}.
    \begin{Corollary} [Corollaries of~\autoref{thm:altfree}]
        $($i$)\,\ f(n,n)\geq n^{\log n}$, $($ii$)\,\ f^{\pl}(n,n^2-n) \geq n^{\log n}$.
    \end{Corollary}
    Viewing the construction behind the proof of~\autoref{thm:altfree} this from the treewidth lens, we have the following connection between pathwidth and alternation-free sequences in certain graphs.
    \begin{Corollary} [Corollary of~\autoref{thm:altfree}] \label{con:end2}
        For every $n,\kappa\geq 4$, there is an a graph on $n\kappa$ vertices of pathwidth $\kappa$ for which there exists an alternation-free sequence of length $n^{\Omega(\log\kappa)}$.
    \end{Corollary}
    
    \subsection{Alternation-free sequences of words} \label{sec:pathstowords}
    
    We first present an alternation-free sequence of paths for non-planar graphs, and then refine it to obtain another for a related planar graph. Consider the graph $G[n,m]$ with vertex set 
    \[ V = \{(i,j): i=0,1,\ldots,n-1, \mbox{ and } j=0,1,\ldots, m-1\} \cup \{s,t\}.\]
    We partition $V \setminus \{s,t\}$ into layers $L_0,L_1, \ldots, L_{m-1}$ of $n$ vertices each, where the $j$-th layer is
    \[ L_j = \{(i,j): i=0,1,\ldots,n-1\}.\]
    There are edges from vertex $s$ to all vertices in $L_0$, and from all vertices in $L_{m-1}$ to $t$. The remaining edges connect vertices in one layer to the vertices in the next. We will have two version of the graph: a non-planar version and a planar version. Let $\Zn = \{0,1,\ldots, n-1\}$, with addition performed modulo $n$. In the non-planar version, we add all edges from a layer to the next. We refer to the resulting graph as $G^{\npl}[n,m]$:
    \[ E(G^{\npl}) = (\{s\} \times L_0) \cup (L_0 \times L_1) \cup (L_1 \times L_2) \cup \cdots
    \cup (L_{m-2} \times L_{m-1}) \cup (L_{m-1} \times \{t\}).\]
    Thus, for $j=0,1,\ldots,m-2$, the subgraph induced by $L_j\cup L_{j+1}$ is a complete bipartite graph. In the planar version, we connect a vertex in layer $j$ to two vertices in layer $j+1$. We refer to the resulting graph as $G^{\pl}[n,m]$:
    \begin{align*}
    E(G^{\pl}) &= \{ ((i,j), (i+b \bmod{n},j+1): b \in \{0,1\},\, \\
    &\qquad i=0,1,\ldots,n-1, \mbox{ and }
    j=0,1,\ldots,m-2\}.
    \end{align*}
    One can imagine that $G^{\pl}$ is drawn on the surface of a cylinder instead of the surface of a plane (the $(n-1)$-th vertex in layer $L_j$ goes around the surface of the cylinder to the $0$-th vertex in layer $L_{j+1}$). See~\autoref{fig:awesomeness} for a depiction of $G^{\pl}$. In $G^{\npl}$, we may encode $s$-$t$ paths by words in $\Zn^m$: the word $\sigma=(\sigma_0,\sigma_1,\ldots,\sigma_{m-1}) \in \Zn^m$ corresponds to the path \[ p_\sigma = (s,(i_0 , 0), (i_1,1), \ldots, (i_{m-1},{m-1}), t),\]
    where $i_0 = \sigma_0$, and $i_{j+1}=i_j+\sigma_{j+1} \bmod{n}$, for $j=0,1,\ldots,m-2$.
    Similarly, we associate words $\tau \in \{0,1\}^m$ with paths $p_\tau$ in $G^{\pl}$. We define alternation-free sequences of words, and observe that the corresponding paths are alternation-free. By showing long alternation-free sequences of words, we establish the existence of long alternation-free sequences of paths.
    
    \begin{Definition}[Word]
        Let $\Zn$ denote the set $\{0,1,2,\ldots,n-1\}$ where addition is performed modulo $n$. Let $\Zn^m$ denote the set of words over $\Zn$ of length $m$. For a word $\sigma \in \Zn^m$ and $i \in \{0,1,\ldots,m-1\}$, let $\sigma[i]$ denote the $i$-th element of $\sigma$; let $\sigma[i:j]$ denote the subword $(\sigma[i],\sigma[i+1], \ldots, \sigma[j-1])$. For $\sigma \in \Zn^k$, let $|\sigma|_1$ denote the sum $($in $\Zn)$ of its elements. That is, $|\sigma|_1 = \sum_{i=0}^{k-1} \sigma[i] \bmod{n}$. Given a word $\sigma \in \Zn^m$ and $j \in \Zn$, let $\sigma \downarrow j$ be the word $\mu \in \Zn^{2m}$ obtained from $\sigma$ by inserting $j$ after each symbol of $\sigma$. That is, if $\mu=\sigma\downarrow j$, then $\mu[2i] = \sigma[i]$ and $\mu[2i+1] = j$, for $i=0,1,\ldots,m-1$. Let $S \in (\Zn^m)^\ell$ be a sequence of words $($each word is in $\Zn^m)$, and $S\downarrow j= (\sigma\downarrow j: \sigma \in S)$ be the sequence obtained after performing such an insertion on every word of $S$. For instance, if $\sigma=(7\ 2\ 6\ 2)$, then $\sigma\downarrow 3 = (7\ 3\ 2\ 3\ 6\ 3\ 2\ 3)$.
    \end{Definition}
    
    \begin{Definition}[Alternation-free sequence of words] \label{def:alternation-free}
        Let $S$ be sequence of $\ell$ words from $\Zn^m$, that is, $S=(\sigma_1,\sigma_2,\ldots,\sigma_{\ell-1})\in(\Zn^m)^\ell$. We say that $S$ has an alternation at $(a,b,c)$ between $(u,v)$, where $0 \leq a < b < c \leq \ell-1$ and $0 \leq u < v \leq m$, if
        \begin{itemize}
            \item $|\sigma_a[0:u]|_1 = |\sigma_b[0:u]|_1 = |\sigma_c[0:u]|_1$;
            \item $v=m$ or $(|\sigma_a[0:v]|_1 = |\sigma_b[0:v]|_1 = |\sigma_c[0:v]|_1)$;
            \item $\sigma_a[u:v] = \sigma_c[u:v] \neq \sigma_b[u:v]$.
        \end{itemize}
        Note that in every such alternation, either $v=m$ or $v-u \geq 2$. If $S$ has no alternation, then we say that $S$ is an alternation-free sequence of words.
    \end{Definition}
    
    \begin{Proposition}[Paths from sequences]
    If $S = (\sigma_i: i=0,1,\ldots,\ell-1) \in (\Zn^m)^\ell$ is an alternation-free sequence of words,
    then $(p_{\sigma_i}: i=0,1,\ldots, \ell-1)$ is an alternation-free sequence of paths in $G^{\npl}_m$. Similarly, if $T = (\tau_i: i=0,1,\ldots,\ell-1) \in (\{0,1\}^m)^\ell$ is an alternation-free sequence of words, then $(p_{\tau_i}: i=0,1,\ldots, \ell-1)$ is an alternation-free sequence of paths in $G^{\pl}_m$.
    \end{Proposition}
    
    \begin{proof}
    Straightforward. Note that the case $v=m$ in the second condition of~\autoref{def:alternation-free} is used to verify that there is no alternation involving pairs of vertices of the form $(u,t)$.
    \end{proof}
    
    Thus, we can now focus on creating alternation-free sequences of words.
    
    \begin{figure}
    \begin{center}
    \begin{tikzpicture} [scale=0.97]

    \tikzset{
    pics/.cd,
    disc/.style = {
    code = {
    \path [top color = teal!50, bottom color = white] (0,0) ellipse [x radius = 2, y radius = 2/3];
    \path [left color = teal!50, right color = teal!50, middle color = teal!10] (-2,0) -- (-2,-\gap) arc (180:360:2 and 2/3) -- (2,0) arc (360:180:2 and 2/3);
    
    \svertex(a) at (-2,0) [fill=\vertcolor, \vertcolor] {};
    \svertex(d) at (2,0) [fill=\vertcolor, \vertcolor] {};
    \svertex(c) at (-60:2 and 2/3) [fill=\vertcolor, \vertcolor] {};
    \svertex(b) at (-120:2 and 2/3) [fill=\vertcolor, \vertcolor] {};
    \svertex(e) at (60:2 and 2/3) [fill=\vertcolor, \vertcolor] {};
    \svertex(f) at (120:2 and 2/3) [fill=\vertcolor, \vertcolor] {};
    
    \draw[line width=\y cm, black, opacity=0.5] (a) to[out=-75,in=150] ($(b)+(0,-\x)$);
    \draw[line width=\y cm, black, opacity=0.5] (b) to[out=-30,in=175] ($(c)+(0,-\x)$);
    \draw[line width=\y cm, black, opacity=0.5] (c) to[out=-5,in=185] ($(d)+(0,-\x)$);
    \draw[line width=\y cm, black, opacity=0.5] (b) to ($(b)+(0,-\x)$);
    \draw[line width=\y cm, black, opacity=0.5] (c) to ($(c)+(0,-\x)$);
    
    \svertex at ($(b)+(0,-\x)$) [fill=\vertcolor, \vertcolor] {};
    \svertex at ($(c)+(0,-\x)$) [fill=\vertcolor, \vertcolor] {};
    \svertex at ($(d)+(0,-\x)$) [fill=\vertcolor, \vertcolor] {};

    }
    },
    }

    \def \gap {0.6};
    \def \x {0.6};
    \def \y {0.01};
    \def \z {0.2};
    \def \zz {0.15};
    \def \o {0.35};
    \def \vertcolor {teal};
    \def \darkcolor {teal};
    \def \msb {darkgray};
    \def \osb {darkgray};
    \def \lsb {darkgray};
    \def \xpos {6.4};
    \def \brc {1};

    \path (0,0) pic {disc} (0,\x) pic {disc} (0,2*\x) pic {disc} (0,3*\x) pic {disc} (0,4*\x) pic {disc} (0,5*\x) pic {disc} (0,6*\x) pic {disc} (0,7*\x) pic {disc} (0,8*\x) pic {disc} (0,9*\x) pic {disc} (0,10*\x) pic {disc} (0,11*\x) pic {disc} (0,12*\x) pic {disc} (0,13*\x) pic {disc} (0,14*\x) pic {disc} (0,15*\x) pic {disc} (0,16*\x) pic {disc} (0,17*\x) pic {disc} (0,18*\x) pic {disc} (0,19*\x) pic {disc} (0,20*\x) pic {disc} (0,21*\x) pic {disc} (0,22*\x) pic {disc} (0,23*\x) pic {disc} (0,24*\x) pic {disc} (0,25*\x) pic {disc} (0,26*\x) pic {disc} (0,27*\x) pic {disc} (0,28*\x) pic {disc} (0,29*\x) pic {disc} (0,30*\x) pic {disc} (0,31*\x) pic {disc} (0,32*\x) pic {disc} (0,33*\x) pic {disc} (0,34*\x) pic {disc};

    \draw[line width=\z cm, red] (-1,33*\x) to[out=-30,in=175] (1,32*\x);
    \draw[line width=\z cm, red] (1,32*\x)--(1,28*\x);

    \draw[line width=\z cm, red] (1,28*\x) to[out=-10,in=-170] (2,28*\x);
    \draw[line width=\zz cm, red, opacity=\o, densely dashed] (2,28*\x) to[out=170,in=10] (1,28*\x);
    \draw[line width=\zz cm, red, opacity=\o, densely dashed] (1,28*\x) to[out=-175,in=30] (-1,27*\x);
    \draw[line width=\zz cm, red, opacity=\o, densely dashed] (-1,27*\x) to[out=-150,in=75] (-2,25*\x);
    \draw[line width=\z cm, red] (-2,25*\x)--(-2,24*\x);

    \draw[line width=\z cm, red] (-2,24*\x) to[out=-75,in=150] (-1,22*\x);
    \draw[line width=\z cm, red] (-1,22*\x)--(-1,18*\x);

    \draw[line width=\z cm, red] (-1,18*\x) to[out=-30,in=170] (1,17*\x);
    \draw[line width=\z cm, red] (1,17*\x) to[out=-10,in=-170] (2,17*\x);
    \draw[line width=\z cm, red] (2,17*\x)--(2,14*\x);

    \draw[line width=\zz cm, red, opacity=\o, densely dashed] (2,14*\x) to[out=170,in=10] (1,14*\x);
    \draw[line width=\zz cm, red, opacity=\o, densely dashed] (1,14*\x)--(1,10*\x);
    \draw[line width=\zz cm, red, opacity=\o, densely dashed] (1,10*\x) to[out=-175,in=30] (-1,9*\x);
    \draw[line width=\zz cm, red, opacity=\o, densely dashed] (-1,9*\x) to[out=-150,in=75] (-2,7*\x);

    \draw[line width=\z cm, red] (-2,7*\x) to[out=-75,in=150] (-1,5*\x);
    \draw[line width=\z cm, red] (-1,5*\x) to[out=-30,in=175] (1,4*\x);
    \draw[line width=\z cm, red] (1,4*\x)--(1,3*\x);

    \draw[line width=\z cm, red] (1,3*\x) to[out=-10,in=-170] (2,3*\x);
    \draw[line width=\z cm, red] (2,3*\x)--(2,-\x);

    \dvertex(out1) at (-2,-\x) [fill=\darkcolor, \darkcolor] {};
    \dvertex(out2) at (2,-\x) [fill=\darkcolor, \darkcolor] {};

    \foreach \i in {0,5,...,35}
    {
        \begin{scope}[shift={(0,\i*\x)}]{ \dvertex at ($(out1)+(0,\i*\x)$) [fill=\darkcolor, \darkcolor] {}; \dvertex at ($(out2)+(0,\i*\x)$) [fill=\darkcolor, \darkcolor] {}; }\end{scope}
    }

    \foreach \i in {-1,4,...,34}
    {
        \begin{scope}[shift={(0,\i*\x)}]{ \dvertex at (-60:2 and 2/3) [fill=\darkcolor, \darkcolor] {}; \dvertex at (-120:2 and 2/3) [fill=\darkcolor, \darkcolor] {}; }\end{scope}
    }

    \begin{scope}[shift={(0,34*\x)}]{ \dvertex at (60:2 and 2/3) [fill=\darkcolor, \darkcolor] {}; \dvertex at (120:2 and 2/3) [fill=\darkcolor, \darkcolor] {}; }\end{scope}

    \fill[fill=\lsb,rounded corners](\xpos,\x) rectangle +(\x,\x);
    \node[white] at (\xpos+\x/2,\x+\x/2) {\LARGE{$1$}};
    \node at (\xpos-1.8,\x+\x/2) {\LARGE{$(1\ 0\ 0\ 0\ 0) = $}};

    \fill[fill=\osb,rounded corners](\xpos,6*\x) rectangle +(\x,\x);
    \node[white] at (\xpos+\x/2,6*\x+\x/2) {\LARGE{$4$}};
    \node at (\xpos-1.8,6*\x+\x/2) {\LARGE{$(1\ 1\ 1\ 1\ 0) = $}};

    \fill[fill=\lsb,rounded corners](\xpos,11*\x) rectangle +(\x,\x);
    \node[white] at (\xpos+\x/2,11*\x+\x/2) {\LARGE{$1$}};
    \node at (\xpos-1.8,11*\x+\x/2) {\LARGE{$(1\ 0\ 0\ 0\ 0) = $}};

    \fill[fill=\msb,rounded corners](\xpos,16*\x) rectangle +(\x,\x);
    \node[white] at (\xpos+\x/2,16*\x+\x/2) {\LARGE{$2$}};
    \node at (\xpos-1.8,16*\x+\x/2) {\LARGE{$(1\ 1\ 0\ 0\ 0) = $}};

    \fill[fill=\lsb,rounded corners](\xpos,21*\x) rectangle +(\x,\x);
    \node[white] at (\xpos+\x/2,21*\x+\x/2) {\LARGE{$1$}};
    \node at (\xpos-1.8,21*\x+\x/2) {\LARGE{$(1\ 0\ 0\ 0\ 0) = $}};

    \fill[fill=\osb,rounded corners](\xpos,26*\x) rectangle +(\x,\x);
    \node[white] at (\xpos+\x/2,26*\x+\x/2) {\LARGE{$4$}};
    \node at (\xpos-1.8,26*\x+\x/2) {\LARGE{$(1\ 1\ 1\ 1\ 0) = $}};

    \fill[fill=\lsb,rounded corners](\xpos,31*\x) rectangle +(\x,\x);
    \node[white] at (\xpos+\x/2,31*\x+\x/2) {\LARGE{$1$}};
    \node at (\xpos-1.8,31*\x+\x/2) {\LARGE{$(1\ 0\ 0\ 0\ 0) = $}};

    \foreach \i in {-1,4,...,29}
    {
        \draw [decorate,decoration={brace,amplitude=15pt}] (\xpos-\brc-3.1,\i*\x+5*\x-0.1) -- (\xpos-\brc-3.1,\i*\x+0.1);
        \draw [decorate,decoration={brace,amplitude=15pt}] (\xpos+\x+\brc-0.2,\i*\x+0.1) -- (\xpos+\x+\brc-0.2,\i*\x+5*\x-0.1);
    }

    \foreach \i in {0,3,...,21}
        \foreach \j in {0,1,...,5}
            \dvertex (ab\j\i) at (\j*0.6+8.1,\i-\x) [fill=\darkcolor, \darkcolor] {};
        
    \draw[line width=\z cm, red] (ab30)--(ab23)--(ab46)--(ab39)--(ab112)--(ab015)--(ab218)--(ab121);

    \end{tikzpicture} \caption{A path in $G^{\pl}$ (left) and $G^{\npl}$ (right) for $n=6$, along with their corresponding words. $s$ is connected to all vertices in the top layer and $t$ is connected to all vertices in the bottom layer.}\label{fig:awesomeness}
    \end{center}
    \end{figure}
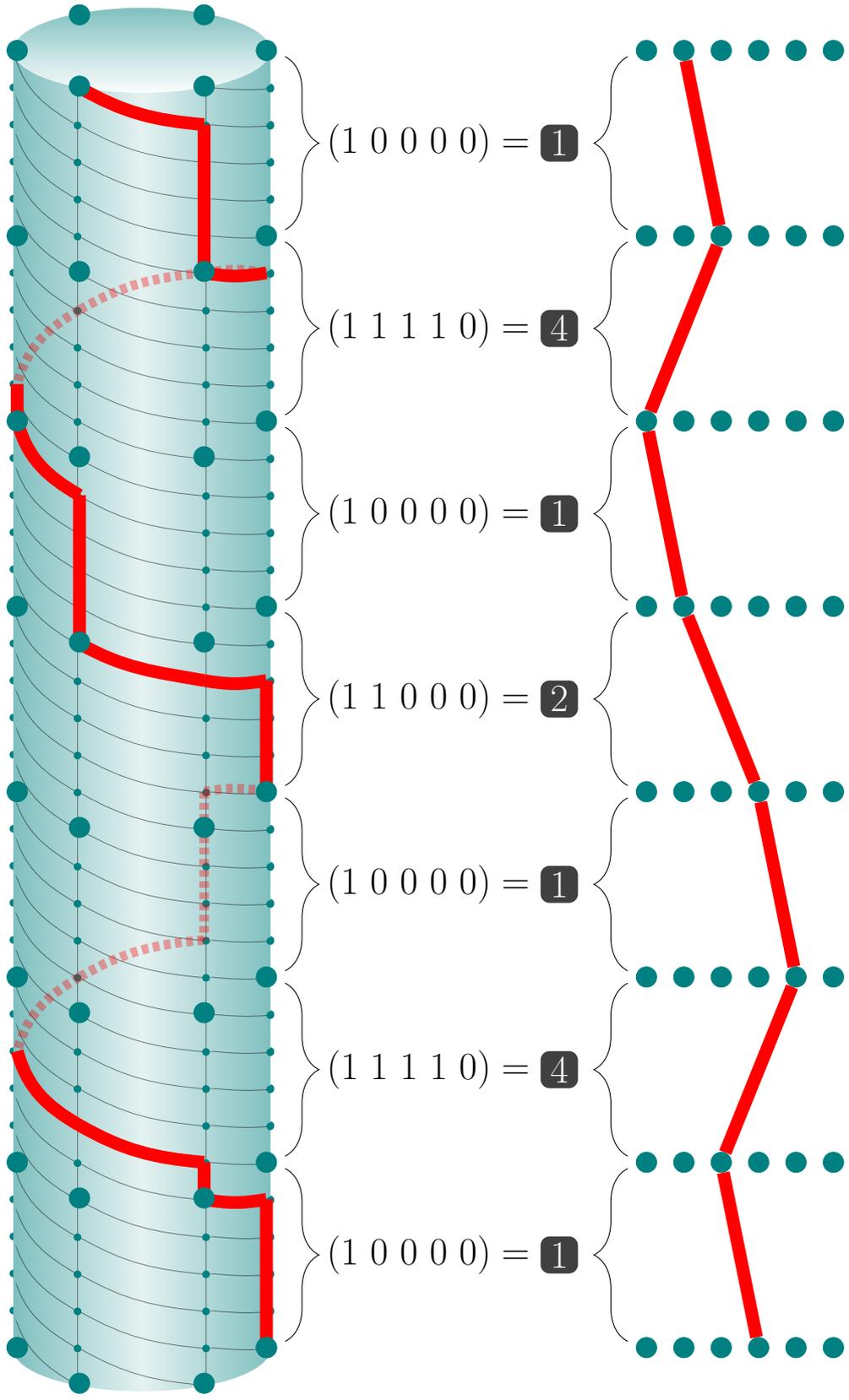
    
    \subsection{Construction of alternation-free sequences of words} \label{sec:words}
    
    In this section, we will construct two alternation-free sequences $X$ and $\hat{X}$ (each of length $n^\ell$) over $\Zn$ and $\{0,1\}$ respectively. We first describe $X$. The $i$-th word ($i=0,1,\ldots,n^{\ell}-1$) of $X$ is given by $X[i] = (b_0) \downarrow b_1 \downarrow \cdots \downarrow b_{\ell-1}$, where $(i)_n=\sum_{j=0}^{\ell-1} b_j n^j$ is the base $n$ representation of $i$. For example, suppose $n=4$ and $i=114$. Then $X[114] = (2) \downarrow 0 \downarrow 3 \downarrow 1=(2\ 1\ 3\ 1\ 0\ 1\ 3\ 1)$ because 114 is equal to 1302 in base 4.
    
    Binary alternation-free sequences can be viewed as a composition of words over $\Zn$, where we map $i \in \Zn$ to the binary word $\hat{\imath} = 1^i0^{n-1-i} \in \{0,1\}^{n-1}$. Thus $\hat{X}[i]$ is constructed exactly like $X[i]$, but it is represented differently (as a binary word of length $(n-1)2^{\ell-1}$ bits). Continuing with the example in the previous paragraph, we have $\hat{X}[114]=(110\ 100\ 111\ 100\ 000\ 100\ 111\ 100)$. Now we will show that $X$ and $\hat{X}$ are alternation-free.
    
    \begin{Lemma} \label{lm:doubling}
    Suppose $S \in (\Zn^m)^{\ell}$ is an alternation-free sequence of $\ell$ words in $\Zn^m$. Then,
    \begin{enumerate}
    \item[$($a$)$] for all $j\in\Zn$, $S \downarrow j$ is an alternation-free sequence of $\ell$ words in $\Zn^{2m}$;
    \item[$($b$)$] $T= (S\downarrow 0) \circ (S \downarrow 1) \circ \cdots \circ (S \downarrow (n-1))$ is an alternation-free sequence of $n \ell$ words, where each word is in $\Zn^{2m}$.\footnote{Here $\circ$ represents concatenation of sequences.}
   \end{enumerate}
    \end{Lemma}
    \begin{proof}
    For part (a), note that if $S\downarrow j$ has an alternation at $(a,b,c)$ $(0 \leq a < b < c \leq \ell-1)$ between $(u,v)$ $(0 \leq u < v \leq 2m)$, then it is easy to see that $S$ itself has an alternation at $(a,b,c)$, between $(\ceil{u/2},\ceil{v/2})$. Since $S$ is alternation-free, so is $S\downarrow j$.

    For part (b), we use part (a). Suppose $T$ has an alternation at $(a,b,c)$ $(0 \leq a < b < c \leq n\ell-1)$ between $(u,v)$ $(0 \leq u < v \leq 2m)$. If $\sigma_a$ and $\sigma_c$ have the same symbol in their odd positions\footnote{Recall that for every $\sigma\in S\downarrow j$, $\sigma$ has $j$ in all its odd positions.} then $\sigma_a$, $\sigma_b$ and $\sigma_c$ all come from a common segment of $T$ of the form $S \downarrow j$. By part (a), the sequence $S \downarrow j$ is alternation-free. So $T$ has no alternation at $(a,b,c)$ between $(u,v)$.
    
    On the other hand, suppose $\sigma_a$ and $\sigma_c$ have different symbols in their odd positions. Since $\sigma_a[v-1] = \sigma_c[v-1]$, we conclude that $v$ is odd. In particular, $v\neq 2m$ and thus $v-u \geq 2$ (as observed at the end of~\autoref{def:alternation-free}). This means that the interval $\{u,u+1,\ldots,v-1\}$ includes an odd number. Hence $\sigma_a[u:v] \neq \sigma_c[u:v]$, and there is no alternation at $(a,b,c)$ between $(u,v)$. This completes the proof.
    \end{proof}
    \begin{Theorem} For all $\ell \geq 1$,  there is an alternation-free sequence $T$ of $n^\ell$ words in $\Zn^{2^{\ell}}$.
    \end{Theorem}
    \begin{proof}
    We will use \autoref{lm:doubling} and induction on $\ell$. For $\ell=1$, the alternation-free sequence is simply $T=(0,1,2,\ldots,n-1)$, which we think of as a sequence of $n$ words, where each word has one symbol. 
    Suppose $\ell >1$. Let $S$ be sequence of $n^\ell$ words in $\Zn^{2^{\ell-1}}$. Consider the sequence \[T= (S\downarrow 0) \circ (S \downarrow 1) \circ \cdots \circ (S \downarrow (n-1)).\] By \autoref{lm:doubling}, $T$ is an alternation-free sequence of $n\cdot n^{\ell}=n^{\ell+1}$ words in $\Zn^{2\cdot 2^{\ell-1}} = \Zn^{2^{\ell}}$. 
    \end{proof}
    Note that $T$ is equal to the $X$ that we had described earlier.
    
    \subsection{Construction of alternation-free sequences of binary words} \label{sec:binarywords}
    
    Consider the following \emph{unary} encoding, where we map $i \in \Zn$ to the binary word $\hat{\imath} = 1^i0^{n-1-i} \in \{0,1\}^{n-1}$. Let $\Znhat=\{\hat{0},\hat{1},\ldots,\widehat{n-1}\}$. Thus, words in $\Znhat^m$ over this alphabet consist of $m$ symbols, each of which is a binary word of $n-1$ bits. We view such a word as a binary string of length $m(n-1)$ by concatenating the $m$ symbols. Now we will show that the resulting sequence of binary strings is alternation-free.
    \begin{Lemma} \label{lm:binarydoubling}
    Suppose $\hat{S} \in (\Znhat^m)^{\ell}$ is alternation-free sequence of $\ell$ binary words of length $m(n-1)$ each. Then,
    \begin{enumerate}
        \item[$($a$)$] for all $\hat{\jmath}\in\Znhat$, $\hat{S} \downarrow \hat{\jmath}$ is an alternation-free sequence of $\ell$ words in $\{0,1\}^{2m(n-1)}$;
        \item[$($b$)$] $\hat{T}= (\hat{S}\downarrow \hat{0}) \circ (\hat{S} \downarrow \hat{1}) \circ \cdots \circ (\hat{S} \downarrow \widehat{n-1})$ is an alternation-free sequence of $n \ell$ words in $\{0,1\}^{2m(n-1)}$.
   \end{enumerate}
    \end{Lemma}
    \begin{proof}
    For part (a), consider $\hat{S} \downarrow \hat{\jmath}$. A word in this sequence consists of blocks of $n-1$ symbols, where each block can be thought of as an element of $\Znhat$. In particular, all the odd numbered blocks of $\hat{S} \downarrow \hat{\jmath}$ comprise of the word $\hat{\jmath}$. Since the symbols from these odd blocks make the same contribution to the prefix sums of all words, we can suppress them and conclude that $\hat{S} \downarrow \hat{\jmath}$ is alternation-free because $\hat{S}$ is known to be alternation-free. We now make this idea more precise.
    
    Suppose $\hat{S} \downarrow \hat{\jmath} = (\sigma_i: i=0,1,\ldots, \ell-1)$ has an alternation at $(a,b,c)$ ($0 \leq a < b < c \leq \ell-1)$ between $(u,v)$ $(0 \leq u < v \leq 2m(n-1))$. First, let us handle the boundary case $v=2m(n-1)$. Then $u$ cannot be a location in the last block, for it is an odd numbered block and the entire block is identical in all words in $\hat{S}\downarrow \hat{\jmath}$. Let $u=q(n-1) + r$, where $r = u \bmod{n-1}$ and $q < 2m-1$. We conclude that $\hat{S}$ has an alternation at $(a,b,c)$ between $(\ceil{q/2}(n-1)+r,m(n-1))$, contradicting the fact that $\hat{S}$ is alternation-free.
    
    So we may assume that $v < 2m(n-1)$ and $v-u\geq 2$. We may also assume that $(u,v)$ has been chosen so that $v-u$ is minimal. This implies that $\sigma_a[u]=\sigma_c[u]\neq \sigma_b[u]$, and similarly that $\sigma_a[v-1]=\sigma_c[v-1]\neq \sigma_b[v-1]$. In particular, both $u$ and $v-1$ are indices in even numbered blocks. Let $u=q(n-1) + r$, $v-1 = q'(n-1) + r'$, where $r = u \bmod{n-1}$ and $r'=v-1 \bmod{n-1}$. Then, $q$ and $q'$ are even. We conclude that $\hat{S}$ has an alternation at $(a,b,c)$ between $((q/2)(n-1) + r, (q'/2)(n-1) + r')$, contradicting the fact that $\hat{S}$ is alternation-free. This establishes part (a).
    
    For part (b), suppose $\hat{T}$ has an alternation at $(a,b,c)$ ($0 \leq a < b < c \leq n\ell-1)$ between $(u,v)$ $(0 \leq u < v \leq 2m(n-1))$; assume that $(u,v)$ has been chosen so that $v-u$ is minimal. Let $\sigma_a\in\hat{S} \downarrow\hat{\jmath_a}$, $\sigma_c\in\hat{S} \downarrow\hat{\jmath_c}$ for some $\hat{\jmath_a}, \hat{\jmath_c}\in\Znhat$. If $\hat{\jmath_a}=\hat{\jmath_c}$, then part (a) gives us the necessary contradiction. So we may assume that $\hat{\jmath_a}\neq\hat{\jmath_c}$. Note that the contents of the odd numbered blocks are \emph{monotonically non-decreasing} in $\Znhat$. Thus if an entire odd numbered block lies in the range $\{u,u+1,\ldots,v-1\}$, then $\sigma_a[u:v] \neq \sigma_c[u:v]$, and there is no alternation. So we may assume that no entire odd numbered block lies in the range $\{u,u+1,\ldots,v-1\}$.
    
    To complete the proof, we now use two crucial facts about our encoding. First, in the sequence of words $(\hat{0}, \hat{1}, \ldots, \widehat{n-1})$, the bit at position $i$ (for $i=0,1,\ldots,n-1$) starts at $0$ and flips to $1$, never to flip back to $0$ again. This implies that both $u$ and $v-1$ lie in even numbered blocks. Since no entire odd numbered block lies in the range $\{u,u+1,\ldots,v-1\}$, this means that $u$ and $v-1$ lie within the same even numbered block. Second, our encoding has the property that if for two words $\hat{\jmath_1}, \hat{\jmath_2} \in \Znhat$, we have $|\hat{\jmath_1}[u:v]|_1 = |\hat{\jmath_2}[u:v]|_1$, then $\hat{\jmath_1}[u:v] = \hat{\jmath_2}[u:v]$. This implies that $\sigma_a[u:v] = \sigma_b[u:v] = \sigma_c[u:v]$; so there is no alternation. This establishes part (b) and completes the proof.
    \end{proof}
    Note that $\hat{T}$ is equal to the $\hat{X}$ that we had described earlier.
    
    \section{Conclusion}
    
    Our ultimate goal is to understand how parametric shortest path complexity changes with the topology of the graph. One of the reasons Nikolova's conjecture was open is that planar graphs have a small (linear) number of edges, which leads one to (falsely) believe that there are not enough edge weights to assign for the graph to have high parametric complexity. Thus, the number of edges is not the right measure to characterize parametric complexity.
    
    Since graphs with large treewidth have superpolynomial parametric complexity and graphs with constant treewidth have polynomial parametric complexity, treewidth seems to be the right measure. However, an unexplored gap still remains. The following conjecture, in particular, is interesting.
    
    \begin{Conjecture} \label{con:con1}
        For every sufficiently large $n, k$, there is an $n$-vertex graph of treewidth $k$ with parametric complexity $n^{\Omega(\log k)}$.
    \end{Conjecture}
    
    This conjecture seems plausible for two reasons: (i) there is a graph on $n$ vertices of pathwidth $k$ having parametric complexity $n^{\Omega(\log k-\log\log n)}$; (ii) there is an a graph on $nk$ vertices of pathwidth $k$ for which there exists an alternation-free sequence of length $n^{\Omega(\log k)}$ (\autoref{con:end2}). These results suggest that we might be very close to resolving~\autoref{con:con1}.
    
    We used alternation-free sequences as a combinatorial way to view parametric shortest paths. Although Kuchlbauer's counterexample~\cite[Example 3.11]{martina} shows that there are infeasible alternation-free sequences, the following question is interesting: what is the worst-case minimum number of paths that we need to delete from an infeasible alternation-free sequence of paths so that the sequence becomes feasible?
    
    \begin{Conjecture} \label{con:con2}
        There exists a universal constant $\gamma$ $(0<\gamma<1)$ such that, for every graph $G$, if the length of the longest alternation-free sequence of paths in $G$ is $L$, then the parametric complexity of $G$ is at least $L^{\gamma}$.
    \end{Conjecture}
    
    If this conjecture is true, then it would imply that alternation-free sequences bound parametric shortest paths from both below and above. More generally, the structure of a graph would completely determine its parametric complexity. However, all known methods for constructing graphs with large parametric complexity do not use the fact that the shortest paths form an alternation-free sequence, and thus it might require considerable insight to resolve~\autoref{con:con2}.
    
    The total number of different paths in a directed acyclic graph on $n$ vertices can be as high as $\exp(n)$. The original problem on parametric shortest paths considers linear edge weights in one variable, which yields shortest path complexity $n^{\Theta(\log n)}$. We generalized the edge weights to polynomials in one variable (\autoref{thm:complpoly}), and to linear forms in three variables (\autoref{thm:complmult}). In both cases, the upper bounds are only slightly higher than those for univariate linear edge weights, and nowhere near $\exp(n)$. It is thus natural to consider a further generalization to multivariate polynomial edge weights.
    
    \begin{Conjecture} \label{con:con3}
        Let $G$ be an $n$-vertex graph whose edge weights are polynomials in $k$ variables of degree at most $d$. Then the parametric complexity of $G$ is at most $n^{(\log n)^\varepsilon}$, where $\varepsilon=\poly(k,d)$.
    \end{Conjecture}
    
    Finally, it is interesting to explore other function families, and to see if there exists a sequence of graphs $\{G_n\}$, where $G_n$ has $n$ vertices, with a well-defined set of edge weight functions on $\mathbb{R}$ such that the parametric shortest path complexity of $G_n$ is $\exp(n)$.
    
    \subsection*{Acknowledgments}
	
	We would like to thank Jannik Matuschke for introducing us to the problem and for several subsequent discussions, Tulasimohan Molli for helping us with the initial analysis of alternation-free sequences in planar graphs, Martina Kuchlbauer for sharing with us her example of infeasible alternation-free sequences, and Suhail Sherif for helping us simulate the problem using a computer program. We are grateful to Vaishali Surianarayanan for suggesting that we consider multivariate linear forms, and to Hariharan Narayanan for helpful discussions on convex polytopes. Finally we would like to thank the anonymous reviewers of this paper for their suggestions and comments, which helped rectify a few minors errors and improve the overall presentation of the paper. 
    
    
	\bibliographystyle{alpha}
	\bibliography{ThirdPaper.bib}
	
	\appendix
	
	\section{The PRAM lower bounds}
	\label{sec:mulmuley-shah-discussion}
	Mulmuley \& Shah's~\cite{mulmuleyshah} Theorem 1.4 states the following.
    \begin{Claim}
    The Shortest Path Problem cannot be computed in $o(\log n)$ steps on an unbounded fan-in PRAM without bit operations using $\poly(n)$ processors, even if the weights on the edges are restricted to have bit-lengths $O((\log n)^2)$.
    \end{Claim}
    A more precise statement of their result (see also Theorem 4.2.1 of Pradyut Shah's PhD thesis~\cite{pradyutshah}) is the following: There exist constants $\alpha  > 0$ and $\epsilon > 0$, and an explicitly described family of weighted graphs $G_n$ ($G_n$ has $n$ vertices and weights that are $O((\log n)^2)$ bits long), such that for infinitely many $n$, every algorithm on an unbounded fan-in PRAM without bit operations with at most $n^\alpha$ processors requires at least $\epsilon \log n$ steps to compute the shortest $s$-$t$ path in $G_n$. (Their proof yields a constant $\alpha < 1$.)
    
    Our proof of~\autoref{thm:corollarymulmuley}, like Mulmuley \& Shah's proof of the corresponding theorem~\cite[Theorem 1.4]{mulmuleyshah}, is based on the following (see~\cite[Theorem 1.1]{mulmuleyshah}).
    \begin{Theorem}
    \label{thm:mulmuleyshahunboundedPRAM}
    Let $\Phi(n,\beta(n))$ be the parametric complexity of any homogeneous optimization problem where $n$ denotes the input cardinality and $\beta(n)$ the bit-size of the parameters. Then the decision version of the problem cannot be solved in the PRAM model without bit operations in $o(\sqrt{\log \Phi(n,\beta(n))})$ time using $2^{\sqrt{\log \Phi(n, \beta(n))}}$ processors even if we restrict every numeric parameter in the input to size $O(\beta(n))$.
    \end{Theorem}
    A version of~\autoref{thm:mulmuleyshahunboundedPRAM} for \emph{bounded fan-in} PRAMs is established in Mulmuley~\cite[Theorem 3.3]{mulmuley}; Mulmuley \& Shah \cite{mulmuleyshah} state that this theorem is also applicable to unbounded fan-in PRAMs. Unfortunately, no formal justification of this latter claim seems to be available in the literature (see Shah~\cite[Page 36]{pradyutshah} for an informal justification).
    
    \section{The weighted graph matching problem} \label{sec:weighted-graph-matching}
	
	In this section, we will see that the shortest path problem reduces to the \textsc{Weighted Graph Matching} problem. More precisely, given a directed acyclic graph $G$ with non-negative edge weights and two special vertices $s$ and $t$, we will show that a minimum weight perfect matching in $G'$ (a slight modification $G$) can be used to compute an $s$-$t$ shortest path in $G$. We now describe how to construct $G'$ from $G$.
	
	For every vertex $v\in V(G)\setminus\{s,t\}$, replace $v$ by two vertices $v_{\operatorname{in}}$ and $v_{\operatorname{out}}$ and add a $0$-weight edge $(v_{\operatorname{in}},v_{\operatorname{out}})$ between them. This edge is the only outgoing edge from $v_{\operatorname{in}}$ and the only incoming edge to $v_{\operatorname{out}}$. The in-neighbours of $v_{\operatorname{in}}$ are the in-neighbours of $v$, and the out-neighbours of $v_{\operatorname{out}}$ are the out-neighbours of $v$. Call this new graph $G'$.
	
	Let $M$ be a minimum weight perfect matching in $G'$ (note that $G'$ always has a perfect matching). Let $v_{\operatorname{in}}^1$ be the partner of $s$ in $M$. This means that the edge $(v_{\operatorname{in}}^1,v_{\operatorname{out}}^1)$ is not in $M$. Now, let $v_{\operatorname{in}}^2$ be the partner of $v_{\operatorname{out}}^1$ in $M$. This means that the edge $(v_{\operatorname{in}}^2,v_{\operatorname{out}}^2)$ is not in $M$. Carrying this argument forward, we obtain that the edge $(v_{\operatorname{in}}^r,v_{\operatorname{out}}^r)$ is not in $M$ (for some $r$), where $v_{\operatorname{out}}^r$ is the partner of $t$ in $M$.
	
	It is easy to check that the path $(s,v^1,v^2,\ldots,v^r,t)$ is an $s$-$t$ shortest path in $G$ (otherwise, a path of lower weight can be used to obtain a matching of lower weight than $M$). In fact, the cost of this path is precisely the weight of $M$, as all the other edges of $M$ are of weight $0$.
	
	\section{Thin grids: an application of alternation-free sequences} \label{sec:thinvsthick}
    
    In this section, we will see that alternation-free sequences can be used to derive upper bounds on the parametric shortest path complexity for a subclass of planar graphs known as grid graphs. 
    
    \begin{Definition}
        The $p\times q$ directed grid graph, denoted by $\Grid_{p,q}$, is defined as follows.
        \begin{enumerate}
            \item[$($a$)$] $V(\Grid_{p,q})=\{(i,j) : 1\leq i\leq p, 1\leq j\leq q\}$.
            \item[$($b$)$] $((i_1,j_1),(i_2,j_2))\in E(\Grid_{p,q})$\footnote{The ordering of $(i_1,j_1)$ and $(i_2,j_2)$ is important since this is a directed graph.} if and only if $(i_1=i_2$ and $j_2=j_1+1)$ or $(j_1=j_2$ and $i_2=i_1+1)$.
        \end{enumerate}
        In other words, the vertices of $\Grid_{p,q}$ form a 2D lattice, and a vertex is connected to the vertex immediately to its right and the vertex immediately above it.
    \end{Definition}
    Let $\compl^{\gr}(p,q,\beta)$ be the parametric shortest path complexity of $\Grid_{p,q}$ where the bit lengths of the coefficients in the weights of the edges are bounded by $\beta$. The planar graphs that we construct as part of our main result (\autoref{thm:mainresult}) can be remodeled into grid graphs at the expense of a small (polynomial factor) blow-up in size of the graph. Thus, $\compl^{\gr}(n,n,O((\log n)^3))\geq n^{\Omega(\log n)}$. This settles the parametric complexity for \emph{square} grids. We ask the same question for thin \emph{rectangular} grids. These are the graphs $\Grid_{p,q}$ with $p\ll q$. Note that $\compl^{\gr}(1,n,\infty)\leq 1$ and $\compl^{\gr}(2,n,\infty)\leq n$ trivially. The problem becomes nontrivial for $3\times n$ grids. We have the following result.
    
    \begin{Theorem} \label{thm:thingrids}
        $\compl^{\gr}(3,n,\infty) \leq 5n$.
    \end{Theorem}
    
    \begin{proof} Our proof is via on an upper bound on the maximum length of an alternation-free sequence of paths in $\Grid_{3,n}$. Now $\Grid_{3,n}$ has $3$ rows and $n$ columns; let the vertices in its middle row be $\{v_1,v_2,\ldots,v_n\}$, arranged in increasing order of their distance from $s$. Our proof strategy is as follows. Given an alternation-free sequence of paths $\paths$ in $\Grid_{3,n}$, we will assign one $v_i$ to each path in $\paths$ (different paths may be assigned the same $v_i$). Then we will show that each $v_i$ can be assigned to at most $5$ paths in $\paths$, thus proving an upper bound of $5n$ on the length of $\paths$.
    
    Since every path from $s$ to $t$ must pass through the middle row, an $s$-$t$ path may be defined by the two vertices it uses to enter and leave the middle row. More formally, for $1\leq i\leq j\leq n$, let $P(i,j)$ be the path from $s$ to $t$ in which $v_i$ is the first vertex of the middle row that lies on $P(i,j)$ and $v_j$ is the last vertex of the middle row that lies on $P(i,j)$. Using this notation, let the alternation-free sequence be $\paths=(P(i_1,j_1),P(i_2,j_2),\ldots,P(i_T,j_T))$. We will prove that $T\leq 5n$.
    
    We now describe how we assign a middle row vertex to each path in $\paths$. For this, we will compare the $k$-th path of $\paths$ with all earlier paths of $\paths$ as follows. For each $k\in\{1,2,\ldots,T\}$, consider the maximum $r$ ($1\leq r\leq k-1$) such that $[i_r,j_r]\cap[i_k,j_k]\neq\emptyset$. Three cases arise.
    
    \begin{enumerate}
        \item[(a)] If no such $r$ exists, then assign $v_{i_k}$ to $P(i_k,j_k)$.
        \item[(b)] If $i_r\neq i_k$, then assign $v_\ell$ to $P(i_k,j_k)$, where $\ell=\max\{i_r,i_k\}$.
        \item[(c)] If $(i_r=i_k$ and $j_r\neq j_k)$, then assign $v_\ell$ to $P(i_k,j_k)$, where $\ell=\min\{j_r,j_k\}$.
    \end{enumerate}
    
    First, note that these are the only possible cases. If case (a) is false (that is, an $r$ does exist), then at least one out of cases (b) or (c) is true, since all the paths in $\paths$ are distinct.
    
    The crucial observation now is that, in $P(i_k,j_k)$, either the vertex $v_\ell$ appears for the first time in $\paths$ (case (a))\footnote{That is, $v_\ell$ is not part of $P(i_r,j_r)$, for all $1\leq r\leq k-1$.}, or the incoming edge to $v_\ell$ has changed since its \emph{most recent} occurrence in $\paths$ (case (b)), or the outgoing edge from $v_\ell$ has changed since its \emph{most recent} occurrence in $\paths$ (case (c)). Fix a middle row vertex $v_m$. Clearly, $v_m$ can appear for the first time in $\paths$ at most once. Also, once $v_m$ has appeared in $\paths$, the incoming and outgoing edges of $v_m$ in later paths of $\paths$ can each change at most two times (see~\autoref{cl:middlelayer} below). Thus, $v_m$ can be assigned to at most $5$ different paths in $\paths$. Summing over all choices of $v_m$, we get $|\paths|=T\leq 5n$.
    \end{proof}
    
    \begin{Claim} \label{cl:middlelayer}
    Let $v_m$ and $\paths$ be as defined in the proof of~\autoref{thm:thingrids}. Then the incoming and outgoing edges of $v_m$ in $\paths$ can each change at most two times.
    \end{Claim}
    \begin{proof}
    We will show that the incoming edge to $v_m$ in $\paths$ can change at most two times. Let $\pred_r(v_m)$ be the predecessor of $v_m$ on the path $P(i_r,j_r)$. Since $v_m$ has in-degree $2$ (for $m>1$), $\pred_r(v_m)$ is either $v_{m-1}$ or $x$, for some vertex $x$ in the first row of $\Grid_{3,n}$. Note that the edge $(x,v_m)$ fixes the $s$-$v_m$ subpath, and changing the incoming edge of $v_m$ in a later path in $\paths$ amounts to abandoning that $s$-$v_m$ subpath. Therefore, the edge $(x,v_m)$ does not occur in any subsequent path in $\paths$. Let us now make this argument formal.
    
    Suppose there exist four paths $P(i_a,j_a), P(i_b,j_b), P(i_c,j_c), P(i_d,j_d)$ in $\paths$ with $a<b<c<d$ such that $\pred_a(v_m)=\pred_c(v_m)=v_{m-1}$ and $\pred_b(v_m)=\pred_d(v_m)=x$. This means that the incoming edge to $v_m$ has changed \emph{three times}. Since there is a unique path from $s$ to $x$ in $\Grid_{3,n}$, we have $P(i_b,j_b)[s,v_m]=P(i_d,j_d)[s,v_m]\neq P(i_c,j_c)[s,v_m]$, implying that $\paths$ is not alternation-free, which is a contradiction.
    
    It can also be shown that the outgoing edge from $v_m$ in $\paths$ can change at most two times. We skip the proof because it is along similar lines.
    \end{proof}
    
    It is not known if $\compl^{\gr}(4,n,\infty) \leq O(n)$. However, a simple induction on the grid size shows the following generalization of~\autoref{thm:thingrids}.
    
    \begin{Theorem} \label{thm:lasttheorem}
        For $3\leq p\leq q$, we have $\compl^{\gr}(p,q,\infty)\leq O(q(\log q)^{p-3})$.
    \end{Theorem}
    
    \begin{proof}
    The proof is by induction on $p$. For the base case ($p=3$),~\autoref{thm:thingrids} implies that $\compl^{\gr}(3,q,\infty)\leq O(q)$. For the inductive case, fix a value of $p$ (where $4\leq p\leq q$), and assume that $\compl^{\gr}(p',q,\infty)\leq O(q(\log q)^{p'-3})$ for all $3\leq p'<p$. Now $\Grid_{p,q}$ has $p$ rows and $q$ columns; let the vertices in its $\ceil{\frac{q}{2}}$-th column be $\{u_1,u_2,\ldots,u_p\}$, arranged in increasing order of their distance from $s$. Our proof strategy is as follows. Let $\paths$ be the longest alternation-free sequence of paths in $\Grid_{p,q}$. We will partition $\paths$ into $p$ alternation-free subsequences\footnote{A subsequence of an alternation-free sequence is also alternation-free.} $\paths_1, \paths_2, \ldots, \paths_p$, and provide an upper bound for each. The sum of these $p$ upper bounds is clearly an upper bound on $|\paths|$.
    
    Let $\paths=\paths_1\cupdot\paths_2\cupdot\cdots\cupdot\paths_p$, where the sequence of paths in each $\paths_i$ respects its original ordering in $\paths$. The partitions are defined as follows. For each path $P\in\paths$, we have $P\in\paths_i$ if and only if $u_i$ is the first vertex of the $\ceil{\frac{q}{2}}$-th column that lies on $P$. For each $\paths_i$, we have \[|\paths_i|\leq \compl^{\gr}\left(i,\ceil{\frac{q}{2}},\infty\right)+\compl^{\gr}\left(p-i+1,q-\floor{\frac{q}{2}}+1,\infty\right).\] We are now ready to provide an upper bound for $\compl^{\gr}(p,q,\infty)$.
    \begin{align*}
        \compl^{\gr}(p,q,\infty)=|\paths|=\sum_{i=1}^{p}|\paths_i|
        &\leq \sum_{i=1}^{p}\left(\compl^{\gr}\left(i,\ceil{\frac{q}{2}},\infty\right)+\compl^{\gr}\left(p-i+1,q-\floor{\frac{q}{2}}+1,\infty\right)\right)\\
        &\leq 2\sum_{i=1}^{p}\compl^{\gr}\left(i,q-\floor{\frac{q}{2}}+1,\infty\right)\\
        &\leq \underset{\mathrm{Term\ I}}{\underbrace{2\,\compl^{\gr}\left(p,q-\floor{\frac{q}{2}}+1,\infty\right)}}+\underset{\mathrm{Term\ II}}{\underbrace{2\sum_{i=1}^{p-1}\compl^{\gr}\left(i,q-\floor{\frac{q}{2}}+1,\infty\right)}}.
    \end{align*}
    $\mathrm{Term\ II}$ can be solved by invoking the induction hypothesis and $\mathrm{Term\ I}$ becomes part of the recurrence.
    \begin{align*}
        \compl^{\gr}(p,q,\infty)&\leq 2\,\compl^{\gr}(p,r,\infty)+2\sum_{i=1}^{p-1}\compl^{\gr}(i,r,\infty) \qquad(\text{where } r=q-\floor{q/2}+1)\\
        &\leq 2\,\compl^{\gr}(p,r,\infty)+2c_1\sum_{i=1}^{p-1}r(\log r)^{i-3} \quad\ \,(\text{applying induction; here } c_1 \text{ is a constant})\\
        &\leq 2\,\compl^{\gr}(p,r,\infty)+2c_1r(c_2(\log r)^{p-4}) \quad(\text{the constant } c_2 \text{ handles lower order terms})\\
        &\leq 2\,\compl^{\gr}(p,r,\infty)+O(q(\log q)^{p-4}).
    \end{align*}
    Since $r$ is roughly $q/2$, evaluating this final recurrence gives $\compl^{\gr}(p,q,\infty)\leq O(q(\log q)^{p-3})$.
    \end{proof}
    \paragraph{Remark:} \autoref{thm:lasttheorem} only helps for small values of $p$, that is, when $p\leq \frac{(\log q)^2}{\log \log q}$. For large $p$, the generalized upper bound of Gusfield gives a much better upper bound.
\end{document}